\documentclass[10pt]{article}

\usepackage{amsmath, amssymb, mathrsfs, vmargin, theorem, dsfont, color}
\usepackage[utf8]{inputenc}
\usepackage[T1]{fontenc}
\usepackage{endnotes}
\usepackage{graphicx}
\usepackage[margin=1.25in]{geometry}

\setmarginsrb{30mm}{30mm}{25mm}{30mm}{0pt}{0pt}{0pt}{26pt}

\theoremheaderfont\scshape
\newtheorem{theorem}{Theorem}[section]

\newtheorem{assumption}{Assumption}[section]

\newtheorem{lemma}{Lemma}[section]
\theorembodyfont{\rmfamily}

\newtheorem{remark}{Remark}[section]

\newtheorem{example}{Example}[section]

\newenvironment{proof}[1][]{\noindent\textit{Proof#1.} }{\vskip\baselineskip}

\renewcommand\hat[1]{\widehat{#1}}

\newcommand\qed{\hfill$\Box$}

\newcommand\esssup{\operatorname{ess\;sup}}

\newcommand\bcdot{\ensuremath{%
  \mathchoice%
   {\mskip\thinmuskip\lower0.2ex\hbox{\scalebox{1.5}{$\cdot$}}\mskip\thinmuskip}}%
   {\mskip\thinmuskip\lower0.2ex\hbox{\scalebox{1.5}{$\cdot$}}\mskip\thinmuskip}%
   {\lower0.3ex\hbox{\scalebox{1.2}{$\cdot$}}}%
   {\lower0.3ex\hbox{\scalebox{1.2}{$\cdot$}}}%
}

\numberwithin{equation}{section}

\begin{document}
\title{Utility indifference valuation for non-smooth payoffs\\ with an application to power derivatives\thanks{The authors thank the Finance and Sustainable Development Chair sponsored
by EDF and CACIB for their support. They are also grateful to Umut \c Cetin, Huyên Pham and Anthony Réveillac for fruitful discussions.}}
\author{Giuseppe Benedetti\footnote{CREST (Finance/Insurance Laboratory) and Universit\'e Paris-Dauphine. E-mail: giuseppe.benedetti@ensae.fr.}   \quad   Luciano Campi\footnote{University Paris 13, CREST (Finance/Insurance Laboratory) and FiME. E-mail: campi@math.univ-paris13.fr.}}
\date{\today}
\maketitle

\begin{abstract}
We consider the problem of exponential utility indifference valuation under the simplified framework where traded and nontraded assets are uncorrelated but where the claim to be priced possibly depends on both. Traded asset prices follow a multivariate Black and Scholes model, while nontraded asset prices evolve as generalized Ornstein-Uhlenbeck processes. We provide a BSDE characterization of the utility indifference price (UIP) for a large class of non-smooth, possibly unbounded, payoffs depending simultaneously on both classes of assets. Focusing then on European claims and using the Gaussian structure of the model allows us to employ some BSDE techniques (in particular, a Malliavin-type representation theorem due to \cite{ma.02}) to prove the regularity of $Z$ and to characterize the UIP for possibly discontinuous European payoffs as a viscosity solution of a suitable PDE with continuous space derivatives. The optimal hedging strategy is also identified essentially as the delta hedging strategy corresponding to the UIP. Since there are no closed-form formulas in general, we also obtain asymptotic expansions for prices and hedging strategies when the risk aversion parameter is small. Finally, our results are applied to pricing and hedging power derivatives in various structural models for energy markets.\medskip\\
\textbf{Keywords :} Utility Indifference Pricing, Optimal Investment, Backward Stochastic Differential Equations, Viscosity Solutions, Electricity Markets.\smallskip\\
\textbf{MS Classification (2010) :} 49L25, 49N15, 60H30, 91G80.

\end{abstract}


\section{Introduction}
This paper deals with the pricing and hedging of derivatives in incomplete markets, where the source of incompleteness comes from the fact that some of the assets are assumed not to be traded. As it is well known, such a situation generally prevents from constructing a perfect hedge and therefore to obtain a unique price as a result of classical no-arbitrage arguments (at least when contingent claims also depend on non-traded assets). In the absence of a unique equivalent martingale measure, indeed, arbitrage theory only allows to identify intervals of viable prices, which makes it necessary to develop other criteria to actually choose a unique price. The easiest and most conservative choice would be (for the seller) to pick the super-replicating price, thus eliminating all the risks by transferring to the buyer the entire cost of the incompleteness. Unfortunately this procedure often gives rise to unreasonably high prices which do not usually match with real data, as it is quite unlikely that one counterpart will completely refuse to take any risk at all. For this reason, other paradigms have been introduced in the literature: one example is Local Risk Minimization (see \cite{schweizer.01}) which does not aim at canceling the hedging risk but rather at minimizing it according to some suitable criterion. Another (partial) way out is the idea of introducing in the market some new assets which are correlated to the non-tradable ones and can therefore be exchanged in the hope of improving the quality of the hedge (see \cite{Davis.97}).
Of course when dealing with the optimal balancing of risks, the standard mathematical way to tackle the problem is the introduction of utility functions, which allow to describe in an easy and concise fashion the amount of uncertainty that an agent is willing to bear. This is at the basis of the well established economic principle of the \textit{certainty equivalent}, stating that the price of a claim should be the one that makes the agent indifferent between possessing the claim or its (certain) price. Such a method has the advantage of being both economically sound and mathematically and computationally simple, requiring at most the numerical evaluation of an equation. This procedure, however, does not seem so appropriate when at least some of the assets can be traded on a financial market: in fact, if the agent is in the position of performing some kind of partial hedging, this should be incorporated in the pricing paradigm, and investors can no longer be expected to passively require an equivalent compensation for claims without engaging in any trading activity. This idea is at the heart of the pricing method that we consider in this paper, i.e. \textit{utility indifference pricing}, a subject that has attracted quite a lot of attention in recent years (see Henderson and Hobson's survey \cite{henderson.04}), in particular as a consequence of the important developments in the theory of optimal investment.\medskip\\
In this article we consider a model for traded and nontraded assets, that are supposed to be uncorrelated. This type of model is usually called semi-complete product market model (as in, e.g., \cite{becherer.03}). The prices of traded assets follow a complete multivariate Black-Scholes model, while the prices of non traded ones evolve as generalized Ornstein-Uhlenbeck processes. This is mainly motivated by the recent literature on structural models for electricity markets, which aim at describing electricity prices as a result of the interaction of some underlying structural factors that can be either exchanged on a financial markets (like fuels) or not (like demand and fuel capacities), and which are often supposed to have simple Gaussian dynamics.
\\In our framework the payoff is supposed to be a function of both traded and nontraded assets, contrarily to most of the literature where the payoff depends only on the nontraded assets which are assumed to be correlated to the traded ones, so that one usually works directly with the correlation of the traded assets with the payoff to be hedged (see, for example, \cite{henderson.02}, \cite{becherer.06}, \cite{Ankirchner.07}, \cite{freischw.08}, \cite{imkeller.12}). An exception is \cite{sircar.zari}, where the payoff considered depends on both types of assets in a bidimensional stochastic volatility framework where the payoff is assumed to be smooth and bounded. Relying on correlation can be advantageous in some situations but not, in general, in the context of structural models, where the expressions for correlations usually become quite complex even if the model is relatively simple. In these cases it is often more convenient to avoid the computation of correlation, by leaving the payoff expressed as a function of both traded and nontraded assets (by eventually exploiting their particular structure, for example their independence or Gaussian properties, to simplify the problem).\\
The typical tool that is used to analyse utility indifference prices is the theory of (quadratic) BSDEs, that was first introduced in a similar context by the seminal paper \cite{nicole.00} and which is particularly convenient as it generalizes with no additional effort to a large class of (possibly non-Markovian) settings (for example \cite{becherer.06}). Classical results require, however, boundedness or at least exponential integrability of the claim and they are only capable to identify the optimal hedging strategy when the final claim is bounded. This is a serious drawback if we notice that common payoff functions in structural models for electricity prices are linear functions of geometric brownian motions (wich are neither bounded nor exponentially integrable).
\\The first contribution of this work is therefore to prove the existence of (exponential) utility indifference prices without requiring boundedness or exponential integrability for the payoff, but only using sub- and super-replicability instead. Nonetheless, the question remains of whether we can actually interpret the $Z$-part of the BSDE in terms of the optimal hedging strategy in this case, given in particular that we lack the BMO property that is generally used to verify this (see \cite{imkeller.05}). With this motivation in mind, we proceed to study the regularity and to get some estimates on $Z$, by using the stochastic control representation of the problem or some Malliavin-type formulas for BSDEs in the spirit of \cite{zhang.05} or \cite{ma.02}.
This is why in the second part of the paper we focus on European payoffs, by allowing them in particular to be possibly discontinuous, which is often the case in models aiming to describe regime-changing features. Given our simple Gaussian modeling framework, considering European payoffs leads naturally to a link with PDEs: our second contribution, indeed, is to describe the price as a viscosity solution of a suitable PDE and, most importantly, to prove that the solution is sufficiently regular to possess continuous first derivatives (in space), providing a useful representation for $Z$ which allows to write the candidate optimal hedging strategy in a similar way as the usual delta hedge. This candidate strategy is then proved to be optimal under some growth assumptions on the payoff (which does not, however, need to be bounded). We stress that our approach is crucially based on the fact that the driver in our BSDE is quadratic in the components of $Z$ corresponding to the nontraded assets whereas it is linear in the other components. Since there is in general no hope to solve the PDE explicitly, we also provide asymptotic expansions for the price (adapting a result in \cite{monoyios.12}) and (under some additional regularity) for the optimal hedging strategy. As already mentioned, we finally provide an application to the pricing of power derivatives under a structural modeling framework.\medskip
\\The paper is organized as follows. We introduce the model in Section \ref{Sec:model}, along with the definition of trading strategies and utility indifference prices, by also deriving some bounds and pointing out the connection with the related concept of certainty equivalent. In Section \ref{Sec:UIP_BSDE} we use some results of the theory of optimal investment (due to \cite{imkeller.05} and \cite{OZ.99}) in order to derive a BSDE representation of the price, without the assumption of boundedness or exponential integrability of the claim that are usually encountered in the literature on quadratic BSDEs (for example \cite{Kob.00} or \cite{Briand.07}). In Section \ref{european} we focus on European payoffs and we express the price and the optimal hedge in terms of viscosity solutions of a certain PDE.  Particular attention is devoted to the case of discontinuous payoffs, that we are able to treat by extending some of the techniques found in \cite{zhang.05}. Asymptotic expansions are also derived following essentially the lines of \cite{Davis.97} and \cite{monoyios.12}. In Section \ref{Sec:electricity} we finally present some applications to electricity markets.\medskip\\
\textbf{Some useful notation:} Let $T>0$ be a finite time horizon and let $(\Omega, \mathbb F, P)$ with $\mathbb F = (\mathcal F_t)_{0\leq t\leq T}$ be a filtered probability space satisfying the usual conditions. For any real number $p>0$, we will denote $\mathbb H ^p (\mathbb R^n)$ (resp. $\mathbb H^p_\textrm{loc} (\mathbb R^n)$) the set of all $\mathbb F$-predictable 
$\mathbb R^n$-valued processes $Z = (Z_t)_{0\leq t \leq T}$ such that $E [\int_0 ^T \| Z_t \|^p dt ] < +\infty$ a.s. (resp. $\int_0 ^T \| Z_t \|^p dt < +\infty$). 
\\ For a vector $x$, we denote $x'$ its transpose and $\textrm{diag}(x)$ the diagonal matrix such that $\textrm{diag}(x)_{ii}= x_i$ for all $i$. For a matrix $\alpha$, we denote $\alpha_{i \bcdot}$, $\alpha_{\bcdot j}$ its $i$'th row or $j$'th column and $\alpha^{-n} := (\alpha^{-1})^n$. For any positive integer $d\geq 1$, we denote $0_d$ the $d$-dimensional zero vector.

\section{The model}\label{Sec:model}
We place ourselves on a filtered probability space $(\Omega,\mathbb F=(\mathcal F_t)_{0\leq t\leq T},P)$, where $\mathbb F$ is the natural filtration generated by the $(n+d)$-dimensional Brownian motion $W=(W^S,W^X)$ and satisfying the usual conditions of right-continuity and $P$-saturatedness. Throughout the paper we will use the notation $y^S$ and $y^X$ to distinguish the first $n$ and last $d$ components of a vector $y=(y^S,y^X)$ of size $n+d$. The distinction is useful, as we will see, to separate tradable and non tradable assets. Moreover, we will denote $\mathbb F^S = (\mathcal F_t ^S)_{0\leq t \leq T}$ and $\mathbb F^X = (\mathcal F_t ^X)_{0\leq t \leq T}$ the natural filtrations generated, respectively, by $W^S$ and $W^X$. The notation $E_t$ will denote conditional expectations under $P$ and with respect to the $\sigma$-field $\mathcal F_t$. \medskip\\
\textbf{Tradable assets.} We consider a finite horizon multivariate Black and Scholes market model with $n$ tradable risky assets with dynamics
\begin{equation}	\label{fuels_dynamics}
\frac{dS^i_t}{S^i_t}=\mu_i dt+\sigma_{i \bcdot} dW^S_t,\quad i=1,\ldots ,n
\end{equation}
where $\sigma$ is a $n\times n$ invertible matrix and $\sigma_{i \bcdot}$ denotes its $i$-th row.
We assume for the sake of simplicity that the interest rate is zero.

\begin{remark}
The results of this paper can be easily extended to the case where the drift and the volatilities in the dynamics of the tradable assets $S$ are bounded functions of these assets, i.e. of the form $\mu (S_t)$ and $\sigma (S_t)$. For the sake of simplicity, we will work under the assumption that they are linear as in (\ref{fuels_dynamics}).
\end{remark}
\textbf{Nontradable assets.} Apart from traded assets, we introduce $d$ non traded assets following the (generalized) Ornstein-Uhlenbeck processes
\begin{equation}
dX^i_t=(b_i (t) -\alpha_i X^i_t) dt+\beta_{i\bcdot} dW^{X}_t, \quad i=1,\ldots,d, \label{X_dynamics}
\end{equation}
where $b_i : [0,T] \to \mathbb R$ is a bounded measurable function and the $\beta_{i\bcdot}$ is the $i$-th row of the $d\times d$-dimensional matrix $\beta$. It is important to remark that as they are defined, tradable and non tradable assets are independent. This is a crucial assumption in what follows. From the modeling viewpoint this is pretty natural since the application we have in mind is to energy markets, where the non tradable assets typically are the electricity demand and the power plant capacities, while the tradable ones are the fuels used in the power production process (such as, for instance, gas, oil and coal). \medskip\\
\textbf{Equivalent martingale measures.} If the market filtration were $\mathbb F^S$ (i.e. that generated by $W^S$ only), then the market would be complete and the unique martingale measure $Q^0$ would be defined by the measure change
$$\frac{dQ^0}{dP}=\mathcal E_T(-\theta \cdot W^S),$$
where $\theta=\sigma^{-1}\mu$ and $\mathcal E$ denotes the stochastic exponential. When considering the whole filtration $\mathbb F$, the market is clearly no longer complete and the set $\mathcal M$ of absolutely continuous martingale measures for $S=(S^1,\ldots,S^n)$ is no longer a singleton. As is well known from the literature (see Schweizer's survey \cite{schweizer.01}), the measure $Q^0$, which is called minimal martingale measure (MMM henceforth), still plays an important role for pricing and hedging derivatives. Remark that in our case the elements of $\mathcal M$ are of the form $\zeta_T=\frac{dQ^0}{dP}M_T$, where the process $M$ is nonnegative and satisfies $E[\zeta_T]=1$. The dynamics of $M$ can be written as
\begin{equation}	\label{mart_meas}
dM_t=\eta_tdW^X_t	\quad M_0=1
\end{equation}
for some $\mathbb F$-predictable process $\eta$. The choice $\eta=0$ (i.e. $M=1$) corresponds to the MMM.\\
We will denote $W^{S,0}=W^S+\theta t$, $W^0=(W^{S,0},W^X)$, and $E^0$ the expectation operator under $Q^0$. Notice that Girsanov's theorem clearly implies that $W^0$ is a $(n+d)$-dimensional Brownian motion under this measure.\medskip\\
\textbf{Trading strategies.} In this model, the wealth process of an agent starting from an initial capital $v\in \mathbb R$ and trading in the risky assets $S$ in a self-financing way over the period $[0,T]$ can be written
$$V^v_t(\pi)=v+\int_0^t \pi'_s (\mu ds+\sigma dW^S_s)=v+\int_0^t \pi'_s\sigma (\theta ds+dW^S_s)$$
where $\pi_s$ is a $n\times 1$ vector representing the investor's trading strategy (in euros) at time $s$ and $\mu$ is a column vector containing the $\mu_i$'s. 
We will need to be more precise later about admissibility conditions on strategies. It is then useful to introduce the following sets:
\begin{eqnarray*}\mathcal H &=& \{\pi \in \mathbb H ^2_\textrm{loc}  (\mathbb R^n) : V^0 (\pi) \mbox{ is a }Q-\mbox{supermartingale for all } Q\in\mathcal M_E\}\\
\mathcal H_M &=& \{\pi \in \mathbb H ^2_\textrm{loc}  (\mathbb R^n) : V^0 (\pi) \mbox{ is a }Q-\mbox{martingale for all } Q\in\mathcal M_E\}\\
\mathcal H_b &=& \{\pi \in \mathbb H ^2_\textrm{loc}  (\mathbb R^n) : V^0 (\pi) \mbox{ is uniformly bounded from below by a constant}\},\end{eqnarray*}
where $\mathcal M_E$ denotes the subset of measures in $\mathcal M$ with finite relative entropy. 
\medskip\\
\textbf{Utility indifference pricing.} We will focus our interest in contingent claims which can depend on both tradable and non tradable assets and which satisfy the following assumption.
\begin{assumption} \label{bounds}
The claim $f$ belongs to $L^2(Q^0,\mathcal F_T)$, it is super/sub-replicable, i.e.
$$V^{v_1}_T(\pi_1)\leq f\leq V^{v_2}_T(\pi_2)$$
for some $v_1,v_2\in \mathbb R$ and $\pi_1\in\mathcal H_M$, $\pi_2\in\mathcal H$. The random variables $V^{v_1}_T(\pi_1), V^{v_2}_T(\pi_2)$ lie in $L^1(Q^0,\mathcal F_T)$. 
\end{assumption}

We focus in this paper on the case of exponential utility $U(x)=-e^{-\gamma x}$, $\gamma >0$, and we look at the buying utility indifference price $p^b$ of the claim $f$ as implicitly defined as a solution to
\begin{equation}     \label{UIP_def}
\sup_{\pi} E\left[U\left(V_T^{v-p^b}(\pi)+ f\right)\right]=\sup_{\pi} E\left[U(V_T^v(\pi))\right]
\end{equation}
where $v\in\mathbb R$ is the initial wealth and the supremum is either taken over $\mathcal H$ or $\mathcal H_b$. It is easily seen that under exponential utility the price is independent of the initial agent's wealth.
By Theorem 1.2 in \cite{OZ.99} the suprema in definition \eqref{UIP_def} are unchanged whether the optimizing set is $\mathcal H$ or $\mathcal H_b$, though the maximum will in general be attained in the larger set $\mathcal H$. \\We will call \textit{optimal hedging strategy} and denote it $\Delta$ the difference between the maxima $\hat \pi^f$ and $\hat \pi^0$ in, respectively, the LHS and RHS of \eqref{UIP_def}, i.e. $\Delta = \hat \pi^{f} - \hat \pi^0$.\\
The selling price $p^s$ is defined similarly as the solution to
$$\sup_{\pi} E\left[U\left(V_T^{v+p^s}(\pi)-f\right)\right]=\sup_{\pi} E\left[U(V_T^v(\pi))\right].$$
We start with a simple preliminary result showing how these prices are related to the expected payoff under the MMM (which can also be interpreted as a price under a certain risk minimizing criterion, see \cite{schweizer.01}). The next result can also be found in \cite{hobson.05}, Theorem 3.1 under slightly different assumptions. We provide here another proof which is perhaps a little bit more general as it is only based on duality (without requiring their Assumption 2.2, even though it would be satisfied in our particular context), and which is also useful to compare utility indifference prices with certainty equivalents (see Remark \ref{rmk_certainty}).

\begin{lemma}   \label{pricebound}
It holds that $$v_1\leq p^b\leq E^0[f]\leq p^s\leq v_2,$$
where $v_1,v_2$ are the same as in Assumption \ref{bounds}.
\end{lemma}
\begin{proof}
We start from the well-known duality result (see \cite{OZ.99}, Theorem 1.1):
\begin{equation}    \label{duality}
\sup_\pi E[U(V^{v-p^b}_T(\pi)+ f)]=\inf_{\delta>0} \inf_{\zeta_T\in\mathcal M}\left\{\delta (v-p^b)+\delta E[\zeta_T f]+E[U^*(\delta\zeta_T)]\right\}
\end{equation}
where $\zeta_T=\frac{dQ^0}{dP}M_T$ as in \eqref{mart_meas} and $U^*$ is the conjugate of $U$.
By taking $M=1$ (equivalently, $\eta=0$) we get
\begin{equation*}
\sup_\pi E[U(V^{v-p^b}_T(\pi)+ f)]\leq\inf_{\delta>0} \left\{\delta \left(v-p^b+ E^0\left[f\right]\right)+E\left[U^*\left(\delta\frac{dQ^0}{dP}\right)\right]\right\}.
\end{equation*}
Now by using \eqref{UIP_def} and \eqref{duality} for $f=0$, we get that
\begin{equation}    \label{duality_2}
\inf_{\delta>0} \inf_{\zeta_T\in\mathcal M}\left\{\delta v+E[U^*(\delta\zeta_T)]\right\}\leq\inf_{\delta>0} \left\{\delta \left(v-p^b+ E^0\left[f\right]\right)+E\left[U^*\left(\delta\frac{dQ^0}{dP}\right)\right]\right\}.
\end{equation}
We want to show that the minimizer in the LHS corresponds to the MMM.
Remark now that for each $\delta>0$ and $\zeta_T=\frac{dQ^0}{dP}M_T$ by using convexity of $U^*$ and conditional Jensen's inequality we get
\begin{equation}	\label{duality_ineq}
\begin{split}
E[U^*(\delta\zeta_T)]&=E\left[U^*\left(\delta\frac{dQ^0}{dP}M_T\right)\right]=E\left[E\left[U^*\left(\delta\frac{dQ^0}{dP}M_T\right)|\mathcal F^S_T\right]\right]\\
&\geq E\left[U^*\left(\delta E\left[\frac{dQ^0}{dP}M_T|\mathcal F^S_T\right]\right)\right]=E\left[U^*\left(\delta \frac{dQ^0}{dP} E\left[M_T|\mathcal F^S_T\right]\right)\right]\\
&=E\left[U^*\left(\delta \frac{dQ^0}{dP}\right)\right]
\end{split}
\end{equation}
where we used the fact that $E[M_T | \mathcal F^S_T]=1$ a.s., which can be shown as follows. By defining $N_t=E_t\left[\frac{dQ^0}{dP}\right]=E\left[\frac{dQ^0}{dP} | \mathcal F^S_t\right]$
we have that
$$E[N_T M_T]=1=N_0 M_0$$
since $N_T M_T$ is a martingale measure density for $S$.
Since $S$ and $X$ are independent, the process $M$ in \eqref{mart_meas} is a positive local martingale in the larger filtration $(\mathcal F^S_T\vee \mathcal F^X_t)_{0\leq t\leq T}$, hence a supermartingale, implying in particular
$$E[M_T | \mathcal F^S_T]\leq E[M_0| \mathcal F_T ^S] = 1.$$
If the previous inequality was strict on a set $F\in\mathcal F^S_T$ of strictly positive probability then we would get the contradiction
$$E[N_T M_T]=E\left[N_TE[M_T| \mathcal F^S_T]\right]< M_0 E[N_T]=1.$$
Therefore if we had $E^0[f]-p^b<0$ by using \eqref{duality_2} and the previous argument we would get the contradiction
\begin{equation*}
\inf_{\delta>0}\left\{\delta v+E[U^*(\delta\zeta^0_T)]\right\}<\inf_{\delta>0} \left\{\delta v+E[U^*(\delta\zeta^0_T)]\right\}.
\end{equation*}
This proves $p^b\leq E^0[f]$.
\\Now consider the super-replicating strategy $\pi_2$ for the claim $f$, starting from a given initial capital $v_2$. Since
\begin{equation*}
\sup_{\pi} E\left[U\left(V_T^{v+v_2}(\pi)- f\right)\right]\geq E\left[U\left(V_T^{v}(\pi)+V_T^{v_2}(\pi_2)- f\right)\right]\geq E\left[U (V_T^{v}(\pi))\right]
\end{equation*}
and therefore
\begin{equation*}
\sup_{\pi} E\left[U\left(V_T^{v+v_2}(\pi)- f\right)\right]\geq \sup_{\pi}E\left[U(V_T^{v}(\pi)\right]
\end{equation*}
we deduce that the selling price $p^s$ must verify $p^s\leq v_2$. The other inequalities are obtained by similar arguments.
\qed\end{proof}

Definition \eqref{UIP_def} can be extended to the conditional case by defining the (buying) price $p^b_t$ as the $\mathcal F_t$-measurable r.v. satisfying
\begin{equation}     \label{UIP_def_t}
\esssup_{\pi} E_t\left[U\left(V_T^{v-p^b_t}(\pi)+ f\right)\right]=\esssup_{\pi} E_t\left[U(V_T^v(\pi))\right]
\end{equation}
where the set of admissible strategies is restrained to those starting at $t$. We denote $p^b_0=p^b$. The previous lemma can therefore be slightly generalized to obtain that
\begin{equation}    \label{bounds_conditional}
  V^{v_1}_t(\pi_1)\leq p^b_t\leq p^s_t\leq V^{v_2}_t(\pi_2).
\end{equation}
Generalizing the other bounds to obtain $p^b_t\leq E^0_t[f]$ is a little bit more delicate since the duality results in \cite{OZ.99} are not extended to the conditional primal problem. A partial result can be obtained using, e.g., BSDE-based methods (see our Remark \ref{bound_uip_t}).

\begin{remark}The previous result confirms that utility indifference valuation gives rise to a sort of bid-ask spread and the price computed under the MMM can be interpreted as a mid price.
The fact that utility indifference buying (selling) prices are always higher (lower) than sub(super)-replication prices also justifies their interest.\end{remark}

\begin{remark}  \label{rmk_certainty}
A related pricing method is given by the \textit{certainty equivalent}, which is quite popular in the economic literature and which has been explored by Benth et al. (\cite{Benth.07}) in the context of electricity markets. In that paper, there is no financial market where the investor could possibly trade. This is the one of the main differences with respect to our approach.
The certainty equivalent method provides the same prices as utility indifference evaluation when the payoff is just a bounded function of the nontraded assets. To see this, remark that when the payoff is bounded we can always perform a probability change and write
$$E[U(V^{0}_T(\pi)+ f-p^b)]=E[e^{-\gamma (f-p^b)}] E^{Q^f}[U(V^{0}_T(\pi))]= E^{Q^f}[cU(V^{0}_T(\pi))]$$
with $c>0$ and the change of measure $\frac{dQ^f}{dP}=\frac{e^{-\gamma (f-p^b)}}{E[e^{-\gamma (f-p^b)}]}$ only affecting the nontraded assets. Let $U^*$ denote the conjugate of $U$. By using $(cU)^*(y)=cU^{*}(y/c)$, the definition \eqref{UIP_def}, the duality results \eqref{duality} and \eqref{duality_ineq} we get
$$\inf_{\delta>0}E\left[U^*\left(\delta \frac{dQ^0}{dP}\right)\right]=\inf_{\delta>0}E^{Q^f}\left[(cU)^*\left(\delta\frac{dQ^0}{dP}\right)\right]$$
which becomes
$$\inf_{\delta>0}E\left[U^*\left(\delta \frac{dQ^0}{dP}\right)\right]=E\left[e^{-\gamma (f-p^b)}\right]\inf_{\delta>0}E\left[U^*\left(\frac{\delta}{E\left[e^{-\gamma (f-p^b)}\right]}\frac{dQ^0}{dP}\right)\right],$$
that is trivially satisfied by the certainty equivalent $p^b=-\frac{1}{\gamma}\ln E[e^{-\gamma f}]$.
However, when the payoff does depend on the traded assets (as in the examples of power derivatives given in Section \ref{Sec:electricity}) the two methods can provide completely different results due to the existence of additional investment opportunities offered by some financial market as it is the case in our model. Notice for instance that the certainty equivalent applied to a payoff which is linear in $s$ (uniformly in $x$) can produce an infinite buying or selling price (since geometric Brownian motion does not have all exponential moments), while by the previous lemma utility indifference prices will always by finite, as the payoff is super/sub-replicable.
\end{remark}

\section{Utility indifference pricing via BSDEs}\label{Sec:UIP_BSDE}
In this section we extend to our setting the classical characterization of the utility indifference price of a contingent claim $f$ in terms of the solution of a suitable BSDE. This characterization has to be proved in our framework since we are not assuming boundedness of $f$ nor that it has finite exponential moments, which are the usual conditions imposed in the existing literature. These conditions would not be satisfied in the application to power derivatives that we have in mind (see Section \ref{Sec:electricity}). From now on we will only focus on buying prices, the selling counterpart being easily obtained by symmetry (see Remark \ref{rmk_selling}). \\
The following result shows how the utility indifference price (UIP for short) is linked to the solution of the BSDE\footnote{It can be viewed as an uncoupled FBSDE since traded and nontraded assets entering in $f$ have forward dynamics.}
\begin{equation} \label{BSDE}
Y_t=f-\int_t^T \left(\frac{\gamma}{2}\| Z^X_s\|^2+\mu'\sigma^{-1}Z^S_s\right)ds-\int_t^T Z_s dW_s
\end{equation}
which can also be written under the MMM $Q^0$ in the simpler form
 \begin{equation} \label{BSDE_MMM}
Y_t=f-\int_t^T \frac{\gamma}{2}\| Z^X_s\|^2ds-\int_t^T Z_s dW^0_s .
\end{equation}
We start by assuming that $f$ is bounded. The next step will consist in replacing the boundedness of $f$ with its sub/super-replicability as in Assumption \ref{bounds}.
\begin{lemma}   \label{UIP_BSDE}
Suppose $f$ is bounded. Then $p^b_t=Y_t$, where $(Y,Z)$ is the unique solution of BSDE \eqref{BSDE} satisfying
$$E\left[\sup_{0\leq t\leq T} |Y_t|^2+\int_0^T \|Z_t\|^2dt\right]<\infty .$$
Moreover, the optimal trading strategy is given by $\Delta_t=-\sigma^{-1}Z^S_t$.
\end{lemma}
\begin{proof}
We prove the lemma only in the case $t=0$. The same arguments easily extend to any time $t$. We use the results in \cite{imkeller.05}. By definition of UIP we are allowed to only consider strategies in $\mathcal H_b$, so that the admissibility conditions in \cite{imkeller.05} are satisfied (apart from square integrability, which is not necessary for what follows). By their Theorem 7, the value function in the LHS of \eqref{UIP_def} takes the form
$$-\exp(-\gamma(x-p-\tilde Y_0))$$
where $\tilde Y_0$ is defined by the unique solution $(\tilde Y, \tilde Z)$ to
\begin{equation}	\label{IMK_BSDE}
\tilde Y_t= -f-\int_t^T \tilde Z_s dW_s-\int_t^T g(s,\tilde Z_s)ds
\end{equation}
with
$$g(\cdot,z)=-\frac{\gamma}{2}\mbox{dist}^2\left(z+\frac{(\theta,0)}{\gamma},C\right)+(\theta,0)'z+\frac{1}{2\gamma}\|\theta\|^2 .$$
In our case $C=\mathbb R^n \times \{0_d\}$, with $0_d$ the null vector in $\mathbb R^d$, and $\|\theta\|^2=\mu'\sigma^{-2}\mu$. Thus, we have
$$g(\cdot,z)=-\frac{\gamma}{2}\|(0,z^X)\|^2+\mu'\sigma^{-1} z^S+\frac{1}{2\gamma}\mu'\sigma^{-2}\mu .$$
When the final claim is zero then $\tilde Z$ in \eqref{IMK_BSDE} vanishes so that it simply gives $\tilde Y_0= -\frac{T}{2\gamma}\mu'\sigma^{-2}\mu$. Applying \eqref{UIP_def} we get
$$-\exp\left(-\gamma\left(x-p-\tilde Y_0\right)\right)=-\exp\left(-\gamma\left(x+\frac{T}{2\gamma}\mu'\sigma^{-2}\mu\right)\right)$$
from which we get
$p=-\left(\tilde Y_0+\frac{T}{2\gamma}\mu'\sigma^{-2}\mu\right)=:Y_0$ where $(Y,\tilde Z)$ solves
$$Y_t= f+\int_t^T\left(-\frac{\gamma}{2}\|(0,\tilde Z^X_s)\|^2+\mu'\sigma^{-1} \tilde Z^S_s\right)ds+\int_t^T \tilde Z_s dW_s .$$
The result then follows by defining $Z=-\tilde Z$. The optimal strategy in the LHS of \eqref{UIP_def} is then given by $\sigma^{-1}\tilde Z^S+\frac{1}{\gamma}\sigma^{-2}\mu$,
and the second result follows.
\qed\end{proof}
\begin{remark}
The result can be also easily derived by properly modifying the proof of Lemma 2.4 in \cite{henderson.11}. However that approach requires a BMO property for admissible strategies which we do not assume.
\end{remark}

\begin{remark} \label{bound_uip_t} Notice also that from the representation \eqref{BSDE_MMM}, by using the classical comparison result for quadratic BSDEs, we can also immediately generalize the result of Lemma \ref{pricebound} by obtaining
\begin{equation}    \label{bound_MMM}
p^b_t\leq E^0_t[f], \quad t\in [0,T].
\end{equation}\end{remark}
We now want to show that \eqref{BSDE} still admits a solution when $f$ is possibly unbounded but still satisfies Assumption \ref{bounds}. We insist once more on the fact that the result is not immediately obvious from the standard literature since $f$ does not necessarily possess exponential moments (e.g. if it depends linearly on the final value of some tradable assets as in our examples in Section \ref{Sec:electricity} of the paper).

\begin{lemma}   \label{existence_BSDE}
Under Assumption \ref{bounds} BSDE \eqref{BSDE_MMM} admits a solution.
\end{lemma}

\begin{proof}
We will adapt the arguments in the proof of Proposition 3 in \cite{Briand.07}. Rewrite equation \eqref{BSDE_MMM} as
\begin{equation}    \label{BSDE_g}
Y_t= f+\int_t^T g(Z_s) ds-\int_t^T Z_s dW^0_s ,
\end{equation}
with $g(z)=-\frac{\gamma}{2}\|z^X\|^2$, and denote $f_n=(-n)\vee f \wedge n$, $L_t= E^0_t[|f|]+E^0_t[|V^{v_1}_T(\pi_1)|]$, $L^n_t= E^0_t[|f_n|]+E^0_t[|V^{v_1}_T(\pi_1)|]$ (which are well defined thanks to Assumption \ref{bounds}). Let $(Y^n,Z^n)$ be the minimal bounded solution to \eqref{BSDE_g} where $f$ is replaced by $f_n$ (it exists by \cite{Kob.00}, Theorem 2.3). By \eqref{bounds_conditional} and \eqref{bound_MMM} we have that $|Y_t^n|\leq L_t^n\leq L_t$ for all $n$. Moreover the sequence $(Y^n)_{n\geq 1}$ is nondecreasing by the comparison theorem (see \cite{Kob.00}, Theorem 2.3).
\\Now define
$$\tau_k=\inf\{t\in [0,T]: L_t\geq k\}\wedge T, \quad \inf \emptyset = +\infty.$$
The sequence $Y^n_k(t)=Y^n_{t\wedge \tau_k}$ remains bounded uniformly in $n$ for each $k$. Setting $Z^n_k(t)=Z^n_{t}\mathbf 1_{\{t\leq\tau_k\}}$ we have
\begin{equation}
Y^n_k(t)=Y^n_{\tau_k}+\int_t^T \mathbf 1_{\{s\leq \tau_ k\}} g(Z^n_k(s)) ds-\int_t^T Z^n_k(s) dW^0_s.
\end{equation}
Now we can apply Lemma 2 in \cite{Briand.07} and obtain, for each $k$, a solution $(Y_k,Z_k)$ to the BSDE
\begin{equation}
Y_k(t)=\xi_k+\int_t^{\tau_k}g(Z_k(s)) ds-\int_t^{\tau_k} Z_k(s) dW^0_s
\end{equation}
where $\xi_k=\sup_n Y^n_{\tau_k}$.
Defining $Y_t=Y_1(t)\mathbf 1_{\{t\leq\tau_1\}}+\sum_{k\geq 2} Y_k(t)\mathbf 1_{]\tau_{k-1},\tau_k]}(t)$ and similarly for $Z_t$ we get
\begin{equation}
Y_t=\xi_k+\int_t^{\tau_k}g(Z_s) ds-\int_t^{\tau_k} Z_s dW^0_s
\end{equation}
and the result follows by sending $k$ to infinity.
\qed\end{proof}
We would like now to be able to interpret the solution $Y$ constructed in the previous lemma as the UIP of the claim $f$. We borrow and adapt the next result from \cite{OZ.99}, which gives some sufficient conditions ensuring this property. Those conditions are quite easy to verify in our setting for a large class of contingent claims (see Section \ref{Sec:electricity}), since the independence between tradable and non tradable assets implies a very simple product structure for the set $\mathcal M$ of all absolutely continuous martingale measures for $S$.
\begin{lemma}   \label{convergence}
Let $f$ be a contingent claim satisfying Assumptions \ref{bounds} and let $f_n=(-n)\vee f \wedge n$, $n\geq 1$. If
\begin{equation}	\label{conv}
\sup_{Q\in\mathcal M_E}E^Q[f_n-f]\to 0, \quad \inf_{Q\in\mathcal M_E}E^Q[f_n-f]\to 0
\end{equation}
as $n\to\infty$ then $Y_0=p^b$, where $Y$ solves \eqref{BSDE}. 
\end{lemma}
\begin{proof}
Following the previous proof, we know that $Y^n_0=p^b(f_n)$, the buying UIP of $f_n$, and that $Y^n_0\to Y_0$, where $Y$ solves \eqref{BSDE}. By Proposition 5.1 (iii) in \cite{OZ.99} we know that
$$\sup_\pi E\left[-e^{-\gamma\left(V_T^{v-p(f_n)}(\pi)+f_n\right)}\right]\to \sup_\pi E\left[-e^{-\gamma\left(V_T^{v-Y_0}(\pi)+ f\right)}\right]$$
which implies that $Y_0=p^b$.
\qed\end{proof}
\begin{remark}
Notice that the conditions in \eqref{conv} are automatically satisfied whenever the super/sub-replicating portfolio strategies are $\mathbb F^S$-predictable and the portfolio values $V^{v_1}_T(\pi_1)$ and $V^{v_s}_T(\pi_2)$ are in $L^2(Q^0,\mathcal F_T)$. This follows from the fact that, for any $Q\in\mathcal M_E$, we have
\begin{equation*}
\begin{split}
E^Q[|f_n-f|]&=E^Q[|f_n-f|\mathbf 1_{|f|\geq n}]\leq \|f_n-f\|_{L^2(Q)} Q(|f|\geq n)^{1/2} \\
&\leq \|f_n-f\|_{L^2(Q)} Q(|V^{v_1}_T(\pi_1)|+|V^{v_2}_T(\pi_2)|\geq n)^{1/2}\\
&\leq C \|f\|_{L^2(Q)} Q(|V^{v_1}_T(\pi_1)|+|V^{v_2}_T(\pi_2)|\geq n)^{1/2}\\
&\leq C (\|V^{v_1}_T(\pi_1)\|_{L^2 (Q)} +\|V^{v_2}_T(\pi_2)\|_{L^2(Q)} ) Q(|V^{v_1}_T(\pi_1)|+|V^{v_2}_T(\pi_2)|\geq n)^{1/2}\\
&= C (\|V^{v_1}_T(\pi_1)\|_{L^2 (Q^0)} +\|V^{v_2}_T(\pi_2)\|_{L^2(Q^0)} ) Q^0 (|V^{v_1}_T(\pi_1)|+|V^{v_2}_T(\pi_2)|\geq n)^{1/2}\\
&  \to 0
\end{split}
\end{equation*}
as $n\to\infty$, where $C>0$ is some constant varying from line to line. This will be the case under, e.g., the Assumptions \ref{Ass:f_cont} and \ref{Ass:f_noncont} that we will introduce in the next section.
\end{remark}
We will now focus on European claims. This will allow, under proper assumptions, to get more information about the process $Z$ and therefore on the hedging strategy. In particular, representation results like those found in \cite{ma.02} or \cite{zhang.05} will reveal to be useful to study the continuity of $Z$ and the possibility to express it starting from the spacial (classical) derivatives of the solution of a given partial differential equation. This will also allow to obtain some estimates on $Z$ which permit to interpret it in terms of the optimal hedging strategy under some less restrictive hypotheses than the boundedness of $f$ (which is required in the standard martingale optimality approach of \cite{imkeller.05} to prove a BMO property for $Z$ which is needed to identify it with the hedging strategy).

\begin{remark}	\label{rmk_selling}
We decided to focus the discussion on buying prices, however most of the results can be adapted to selling prices. In particular the usual relation $p^s(f)=-p^b(-f)$ holds between selling and buying prices. The natural candidate for the selling price is the solution to the BSDE
\begin{equation} \label{BSDE_selling}
Y_t=f+\int_t^T \frac{\gamma}{2}\| Z^X_s\|^2ds-\int_t^T Z_s dW^0_s .
\end{equation}
Remark immediately the sort of symmetry with \eqref{BSDE_MMM}. Existence for \eqref{BSDE_selling} can be obtained by following the proof of Lemma \ref{existence_BSDE}, but using the super (instead of sub)-replicating process in Assumption \ref{bounds} as a bound.
Moreover, under the same conditions as in Lemma \ref{convergence} we are able to interpret this solution as the selling price. All the other results still hold for selling prices with minor modifications. In particular, Lemma \ref{L2} finds its analogue in Lemma \ref{L2_selling}. Both are relegated in the Appendix for the sake of readability.
\end{remark}

\section{Pricing and hedging of European payoffs}\label{european}
In this section we will address the problem of computing the utility indifference price and the corresponding optimal hedging strategy of a European contingent claim $f$, which is a function of both tradable and non tradable assets at the terminal date $T$, i.e. we assume (with a slight abuse of notation) that $f=f(S_T,X_T)$ for some measurable function $f:\mathbb R^n _+ \times \mathbb R^d \to \mathbb R$. We will denote $f_{s^i \pm} (s,x)$ and $f_{x^j \pm}(s,x)$ the right/left derivatives of the function $f(s,x)$ with respect to, respectively, $s^i$ and $x^j$ for $i=1,\ldots,n$ and $j=1,\ldots,d$.\\
We will use the notation $A_t=(S_t,X_t)$ when we wish to consider asset processes with no distinction. The standard notation $E_{t,a}$ denotes expectation with respect to $\mathcal F_t$ given that the process $A$ takes the value $a=(s,x)$ at time $t$.
Our goal is to obtain a complete characterization of the optimal hedging strategy $\Delta$ as well as asymptotic expansions of the price of the contingent claim $f$ for a small risk aversion parameter $\gamma$. Using BSDE techniques and thanks to the Markovian framework, we are able to do so for a large class of non-smooth contingent claims. More precisely, we consider the following two kinds of assumptions for $f$.

\begin{assumption}	\label{Ass:f_cont} \emph{(Continuous non-smooth payoffs)}
The payoff function $f$ is continuous and a.e. differentiable. Moreover, $f$ and its right/left derivatives grow polynomially in $s$, uniformly in $x$, i.e.
$$| f(s,x) |+| f_{s^i \pm} (s,x) | + | f_{x^j \pm} (s,x) | \leq C(1+\|s\|^{q}) , \quad (s,x) \in \mathbb R_+ ^n \times \mathbb R^d,$$
for all $i=1,\ldots,n$ and $j=1,\ldots ,d$ and for some $q\geq 1$, where the constant $C$ does not depend on $x$.
\end{assumption}
\begin{assumption}	\label{Ass:f_noncont} \emph{(Discontinuous payoffs)}
The payoff function $f$ is bounded from below and a.e. differentiable. Moreover
\begin{enumerate}
\item $f$ may have finitely many discontinuities only in the $x$-variables and outside these points $f$ is continuously differentiable.
\item Where it exists, its derivative $f_{s^i \pm}(s,x)$ is bounded, and in particular $f_{s^i \pm} (s,x)= O(1/{s^i})$ for $s^i$ large enough, for all $i=1,\ldots,n$ uniformly in $x$.
\item Where it exists, its derivative $f_{x^j \pm} (s,x)$ verifies $| f_{x^j \pm} (s,x) |\leq C(1+\|s\|^{q})$ for all $j=1,\ldots,d$ and some $q\geq 1$, where the constant $C>0$ does not depend on $x$.
\end{enumerate}
\end{assumption}
We see that if we want to treat discontinuous payoffs we need stronger growth assumptions than in the continuous case. In particular the second hypothesis implies a uniform logarithmic growth of $f$ in the traded assets.
The main example we think about in this case is that of a payoff which separates the contributions of traded and nontraded assets in a multiplicative way (see Section \ref{Sec:electricity} for some examples coming from electricity markets).\medskip\\
Since $f=f(A_T)=f(S_T,X_T)$ we can exploit the Markovian setting and look for a solution to \eqref{BSDE} of the form $Y_t=\varphi(t,A_t)$ where $\varphi(t,a)=\varphi(t,s,x)$ solves the PDE
\begin{equation}    \label{PDE1}
\begin{cases}
\mathcal L\varphi-\frac{\gamma}{2}\sum_{j=1}^d (\beta'_{\bcdot j}\varphi_x)^2=0\\
\varphi(T,a)=f(a)
\end{cases}
\end{equation}
where $\beta_{\bcdot j}$ denotes the $j$-th column of the matrix $\beta$ and
$$\mathcal L \varphi=\varphi_t+(b-\alpha x)\varphi_{x}+\frac{1}{2}\sum_{i,j=1}^n\sigma_{i\bcdot}\sigma_{\bcdot j}s^is^j \varphi_{s^is^j}+\frac{1}{2}\sum_{i,j=1}^d\beta_{i\bcdot}\beta_{\bcdot j} \varphi_{x^ix^j}.$$
The PDE above is motivated by a formal application of It\^o's lemma (to be justified later) for the function $\varphi$ and recalling that with sufficient regularity we expect to have $Z^{X,i}=\beta'_{\bcdot i}\varphi_x$, where $\varphi_x$ is the ($d$-dimensional) gradient of $\varphi$ with respect to $x$.
\\Now denote
$$h(q)=\frac{\gamma}{2}\|q\|^2=\sup_{\delta\in\mathbb R^d}\left\{-q\delta-\frac{1}{2\gamma}\|\delta\|^2\right\}, \quad q \in \mathbb R^d.$$
In this way \eqref{PDE1} can be written as
\begin{equation}	\label{PDE3}
\begin{cases}
-\mathcal L\varphi+h(\beta'\varphi_x)=0\\
\varphi(T,a)=f(a).
\end{cases}
\end{equation}
The main result of this section can be summarized in the following theorem.

\begin{theorem}\label{main} Let $f=f(A_T)=f(S_T, X_T)$ be a given European type contingent claim for some measurable payoff function $f:\mathbb R^n _+ \times \mathbb R^d \to \mathbb R$. We have the following properties.
\begin{enumerate}
\item Under Assumption \ref{Ass:f_cont} or \ref{Ass:f_noncont} the buying UIP $\varphi$ of the claim $f$ is a viscosity solution of \eqref{PDE1} on $[0,T)\times \mathbb R^n_+\times\mathbb R^d$, which is also differentiable in all the space variables.
\item 
The optimal hedging strategy is given by $$\Delta_t=-\sigma^{-1}Z^S_t=-\sigma^{-1}\sigma(S_t)\varphi_s(t,A_t),$$
where $(Y,Z)$ is solution to \eqref{BSDE_MMM} and $\sigma(S)$ the $n\times n$ matrix whose $i$-th row is given by $\sigma_{i\bcdot} S^i$.\\
\end{enumerate}
\end{theorem}
The rest of this section is devoted to proving this theorem and deducing some asymptotic expansions of the price and the optimal hedging strategy for a small risk aversion parameter $\gamma$.

\subsection{Proof of the main theorem}
Before giving the technical details, we briefly sketch the main ideas underlying our proofs. Equation \eqref{PDE3} suggests that we can look at our pricing problem as a stochastic control problem with a quadratic cost function: following this intuition, the idea of the proof is to start with a slightly modified reformulation (using some ideas developed in \cite{pham.02}) in which the control space is forced to be compact. When the payoff is regular enough, this trick allows us to prove the existence of a smooth solution to the modified problem, which immediately extends to the original one by using some estimates on the derivatives which do not depend on the size of the control space. When the payoff is continuous but not smooth enough, we will approximate it with a sequence of smooth ones (to which our previous results apply) and study the behavior of prices in the limit: in particular, by using a Malliavin-type representation of the derivatives which does not rely on the regularity of the payoff (which is due to \cite{ma.02}), we are able to prove that the limiting price function remains differentiable in the state variables (though it possibly fails to be $C^{1,2}$). The case of discontinuous payoffs is a little bit more delicate: again the aim is to obtain some estimates on the derivatives which do not depend on the approximating sequence for the payoff, but here we can exploit neither the derivatives of the approximating sequence (which may explode due to the discontinuities) nor the Malliavin-type estimates in \cite{ma.02} and \cite{zhang.05} which do not apply to quadratic BSDEs. We will tackle the problem by performing a suitable change of measure, which however requires stronger assumptions with respect to the case of continuous payoffs.

\subsubsection{An auxiliary problem with compact control space and smooth terminal condition}

We start analyzing \eqref{PDE3} by forcing the space of controls to be compact, in particular by replacing the function $h(q)$ in \eqref{PDE3} by $h^m (q)$ defined as
$$h^m(q)=\sup_{\delta\in\mathcal B^m(\mathbb R^d)}\left\{-q\delta-\frac{1}{2\gamma}\|\delta\|^2\right\}$$
where $\mathcal B^m(\mathbb R^d)$ is the ball in $\mathbb R^d$ centered at zero and of radius $m > 0$. Thus, the PDE we consider in this section is
\begin{equation}	\label{visc_solution}
\begin{cases}
-\mathcal L\varphi^m+h^m(\beta'\varphi^m_x)=0\\
\varphi^m(T,a)=f(a).
\end{cases}
\end{equation}
We also write its associated BSDE
\begin{equation}	\label{BSDE_m}
Y^m_t=f-\int_t^T h^m(Z^{X,m}_r)dr-\int_t^T Z^m_r dW^0_r
\end{equation}
that we will refer to in the sequel. Existence and uniqueness of the solution for this BSDE are guaranteed by classical results in \cite{pardoux.peng.90}, since the generator $h^m$ is a Lipschitz function.
\begin{lemma}		\label{solution_PDE}
Let $m>0$. If $f\in C^{3}$ and $f$ and all its first derivatives have polynomial growth, then there exists a classical solution $\varphi^m$ to \eqref{visc_solution}. If $f$ is only of polynomial growth (and possibly discontinuous), then $\varphi^m$ is characterized as a continuous viscosity solution to \eqref{visc_solution} with continuous first derivatives in all the space variables, which have the representation
\begin{equation}    \label{zhang1}
\varphi_{a}^m(t,a)= E^0_{t,a}\left[f(A_T)N_T-\int_t^T h^m(Z^{X,m}_r)N_r dr\right]
\end{equation}
(where $\varphi_{a}^m$ is to be interpreted as a column vector in $\mathbb R^{n+d}$ containing the derivatives with respect to the traded and nontraded assets)
with
$$N_r=\left(\begin{array}{c}\frac{1}{r-t}\sigma^{-1}(S_t)'(W^S_t-W^S_r)\\ \frac{1}{r-t}\int_t^r \emph{diag}(e^{-\alpha (u-t)})' \beta^{-1} dW^X_u\end{array}\right).$$
Moreover the following stochastic control representation holds:
\begin{equation}	\label{stoch_rep}
\varphi^m(t,a)=\inf_{\delta\in \mathcal A_t^m}E^Q_{t,a}\left[\frac{1}{2\gamma}\int_t^T \|\delta_r\|^2 dr+f(A_T)\right]
\end{equation}
for some auxiliary probability measure $Q$, under which
\begin{equation}	\label{mod_dynamics}
\begin{cases}
\frac{dS^i_t}{S^i_t}=\sigma_{i\bcdot} dW^{S,Q}_t, \quad i=1,\ldots,n \\
dX^i_t=(b_i(t)-\alpha_i X^i_t+\beta_{i\bcdot}\delta_t) dt+\beta_{i\bcdot} dW^{X,Q}_t \quad i=1,\ldots,d
\end{cases}
\end{equation}
where $(W^{S,Q},W^{X,Q})$ is a $n$-dimensional BM under the measure $Q$ and $\mathcal A_t^m$ stands for the class of adapted $\mathbb R^d$-valued controls $\delta_s$ starting from time $t$ and such that $\| \delta_s \|\leq m$.
\end{lemma}
\begin{remark}
Recall that only the dynamics of nontraded assets are touched under the new measure $Q$, while traded assets still evolve as under the MMM $Q^0$.
\end{remark}
\begin{proof} We split the proof into two main steps.\medskip\\
\textbf{Step 1}: The case where $f$ is smooth follows by Theorem 6.2 in \cite{fleming.75} (or Theorem IV.4.3 in \cite{fleming.06}). The reason for introducing the index $m$ comes from the fact that those theorems require that controls must take values in a compact space.\footnote{The lack of uniform parabolicity here can be handled by a standard logarithmic transformation in the tradable assets. Under the new logarithmic variable, however, the payoff will not preserve polynomial growth in general. Therefore the result should first be applied to PDE \eqref{visc_solution} (under the new variable) where the payoff is replaced by $f(s\wedge C,x)$ for some constant $C>0$, then undoing the logarithmic change of variable and letting $C \to \infty$ will get the final result.} The regularity of $\varphi^m$ implies (by an application of It\^o's lemma) that $\varphi^m(t,A_t)=Y^m_t$, where $Y^m$ solves \eqref{BSDE_m} and $Z^{X,m}_t=\beta'\varphi^m_x(t,A_t)$, $Z^{S,m}_t=\sigma(S_t)'\varphi^m_s(t,A_t)$.
\\We need to introduce the tangent process of $A$, $\nabla A$ (see, e.g., equation (2.9) in \cite{ma.02} for a definition), which has the following characterization in our particular case: $(\nabla A_t)_{ii}= S^i_t / S^i_0$ if $i\leq n$, $(\nabla A_t)_{ii}=e^{-\alpha_i t}$ if $n+1\leq i\leq n+d$, $(\nabla A_t)_{ij}=0$ if $i\neq j$. Now, define $\Sigma(S)$ as the $(n+d)\times (n+d)$ matrix composed by $\sigma(S)$ on the upper left side and $\beta$ on the lower right side, being zero everywhere else.
The $n\times n$ matrix $\sigma^{-1}(S)$ coincides with the matrix where the $i$-th column is equal to the $i$-th column of $\sigma^{-1}$ divided by $S^i$. Then $\Sigma^{-1}(S_t)\nabla A_t$ is equal to $\sigma^{-1}(S_0)$ on the upper-left corner and $B^{-1}(t)$ on the lower-right corner, being zero everywhere else. Define the $(n+d)$-dimensional processes
\begin{equation*}
M_r=\int_t^r(\Sigma^{-1}(S_u)\nabla A_u)'dW^0_u=\left(\begin{array}{c}\sigma^{-1}(S_0)'(W^{S,0}_t-W^{S,0}_r)\\ \int_t^r \textrm{diag}(e^{-\alpha u})' \beta^{-1} dW^X_u\end{array}\right)
\end{equation*}
and
\begin{equation}
N_r=\frac{1}{r-t}M'_r(\nabla A_t)^{-1}=\left(\begin{array}{c}\frac{1}{r-t}\sigma^{-1}(S_t)'(W^S_t-W^S_r)\\
\frac{1}{r-t}\int_t^r \textrm{diag}(e^{-\alpha (u-t)})' \beta^{-1} dW^X_u\end{array}\right).\label{def-N}
\end{equation}
Since $h^m$ is a Lipschitz function for all fixed $m\geq 0$, we can apply the results in \cite{ma.02} (in particular Theorem 4.2) to the processes $M$ and $N$ just defined to show that \eqref{zhang1} is true.
Theorem 4.2 in \cite{ma.02} requires uniform parabolicity which is not respected in our case, however again this is not a problem for geometric Brownian motions since only the process $M$ defined above enters in its proof.\medskip\\
\textbf{Step 2}: In order to prove the result for a general (possibly discontinuous) $f$ we can adapt the proof of Theorem 3.2 in \cite{zhang.05} to our framework. In particular, we can take a sequence $f^l$ of smooth functions with bounded first derivatives such that $f^l\to f$ a.e. as $l\to\infty$. Then we have $f^l(A_T)\to f(A_T)$ $Q^0$-a.s. since all the processes have absolutely continuous densities. Then one defines
$$\varphi^{m,l}(t,a)=Y^{m,l}_t=f^l-\int_t^T h^m(Z^{X,m,l}_r)dr-\int_t^T Z^{m,l}_r dW^0_r$$
We have
\begin{equation}    \label{zhang}
\varphi_{a}^{m,l}(t,a)= E^0_{t,a}\left[f^l(A_T)N_T-\int_t^T h^m(Z^{X,m,l}_r)N_r dr\right]
\end{equation}
and with the same arguments as in \cite{zhang.05}, Theorem 3.2 (slightly modified to our multivariate setting) we can also obtain the estimate
\begin{equation}	\label{der_boundT}
\|\varphi_{a}^{m,l}(t,a)\|\leq C\frac{\|a\|^q}{\sqrt{T-t}}.
\end{equation}
for some $q\geq 0$.
Here the constant $C$ does not depend on $l$ but it depends on $m$ through the Lipschitz constant of $h^m$. Applying classical stability results for BSDE established in, e.g., \cite{ma.yong.07}, Theorem 4.4, we have the convergence
$$E^0\left[\sup_{0\leq t\leq T} |Y^{m,l}_t-Y^{m}_t|^2+\int_0^T \|Z^{m,l}_t-Z^{m}_t\|^2dt\right]\to 0$$
as $l\to\infty$, where $(Y^m,Z^m)$ solve \eqref{BSDE_m} (but with a nonsmooth $f$ as terminal condition). We deduce from Lemma 6.2 in \cite{fleming.06} and the estimate \eqref{der_boundT} (which gives uniform convergence on compact subsets of $[0,T)\times \mathbb R^{n+d}$) that $\varphi^{m,l}\to\varphi^{m}$, where $\varphi^m$ the unique viscosity solution of \eqref{visc_solution}, which is continuous except possibly at $T$. Following the last part of Zhang's proof of Theorem 3.2 in \cite{zhang.05} we also obtain that $\varphi^m$ is differentiable and we have
$$\varphi_{a}^{m}(t,a)= E^0_{t,a}\left[f(A_T)N_T-\int_t^T h^m(Z^{X,m}_r)N_r dr\right].$$
It remains to prove that the stochastic representation \eqref{stoch_rep} holds for $\varphi^m$.
Clearly it holds for $\varphi^{m,l}$ as the approximating functions are continuous, so
\begin{equation*}
\varphi^{m,l}(t,a)\leq E^Q_{t,a}\left[\frac{1}{2\gamma}\int_t^T \|\delta_r\|^2 dr+f^l(A_T)\right]
\end{equation*}
for any $\delta\in \mathcal A_t^m$ and therefore
$$\varphi^{m}(t,a)\leq E^Q_{t,a}\left[\frac{1}{2\gamma}\int_t^T \|\delta_r\|^2 dr+f(A_T)\right]$$
by dominated convergence (since $f$ has polynomial growth), and
$$\varphi^{m}(t,a)\leq  \inf_{\delta\in \mathcal A_t^m}E^Q_{t,a}\left[\frac{1}{2\gamma}\int_t^T \|\delta_r\|^2 dr+f(A_T)\right]$$
To obtain the reverse inequality it suffices to note that we can choose $f^l\geq f$.
\qed\end{proof}

\subsubsection{Continuous non-smooth payoffs}
The next step is now to remove the dependence on the parameter $m$ and to characterize the price $\varphi$. We work in this section under Assumption \ref{Ass:f_cont}.
We start with a useful probabilistic characterization of the derivatives of $\varphi^m$ under this assumption (such derivatives exist even if $\varphi^m$ is only a viscosity solution by \ref{solution_PDE}).
\begin{lemma}	\label{der_repr}
Let $m>0$. Under Assumption \ref{Ass:f_cont} we have the following representations:
\begin{equation}	\label{est_der}
\varphi^m_{s^i}(t,a)= E^Q_{t,a}\left[f_{s^i}(A_T)\frac{S^i_T}{S^i_t}\right],  \quad  \varphi^m_{x^j}(t,a)= e^{-\alpha_j(T-t)}E^Q_{t,a}\left[f_{x^j}(A_T)\right]
\end{equation}
for $i=1,\ldots,n$, $j=1,\ldots,d$, where the processes evolve as in \eqref{mod_dynamics} with $\delta=\hat\delta$, the maximizer in $h^m(\beta'\varphi^m_x)$.
\end{lemma}
\begin{proof}
We adapt the arguments in \cite{fleming.06}, Lemma 11.4, to our slightly different framework.
First assume that $f$ is smooth (in the sense of Lemma \ref{solution_PDE}), then there exists an optimal Markov feedback $\hat \delta\in\mathcal A^m_0$ (the one achieving the max in $h^m(\beta'\varphi^m_x)$) such that
$$\varphi^m(t,a)=E^Q_{t,a}\left[\frac{1}{2\gamma}\int_t^T \|\hat\delta_r\|^2 dr+f(A_T)\right]$$
By using the same control but with different initial condition we clearly obtain
$$\varphi^m(t,a+\varepsilon e_i)\leq E^Q_{t,a+\varepsilon e_i}\left[\frac{1}{2\gamma}\int_t^T \|\hat\delta_r\|^2 dr+f(A_T)\right], \quad i= 1, \ldots, n+d.$$
Taking the difference and dividing by $\varepsilon>0$ we get
$$\frac{\varphi^m(t,a+\varepsilon e_i)-\varphi^m(t,a)}{\varepsilon}\leq E^Q_{t}\left[\frac{f(A^{t,a+\varepsilon e_i}_T)-f(A^{t,a}_T)}{\varepsilon}\right], \quad i=1,\ldots, n+d,$$
where for clarity we wrote here $A^{t,a}_T$ to stress that the process starts at time $t$ with value $a$. The polynomial growth property in the traded assets of the derivatives of $f$ allows us to apply dominated convergence (since traded assets have the same dynamics under $Q$ and $Q^0$, see \eqref{mod_dynamics}) to get
$$\varphi^m_{a^i}(t,a)\leq E^Q_{t}\left[f_{a^i}(A^{t,a}_T)\frac{\partial}{\partial a^i}A^{t,a,i}_T\right], \quad i=1,\ldots, n+d.$$
By repeating the argument with $-\varepsilon$ we finally obtain
$$\varphi^m_{a^i}(t,a)= E^Q_{t,a}\left[f_{a^i}(A_T)\frac{\partial}{\partial a^i}A^{i}_T\right]$$
for $i=1,\ldots,n+d$, which gives the result by considering traded and non traded assets separately.
\\The general result follows by considering an approximating sequence $f^l$ as in the proof of Lemma \ref{solution_PDE} and using dominated convergence.
\qed\end{proof}

If the payoff $f$ is sufficiently regular we can immediately remove the dependence on $m$, as is shown in the next result.

\begin{lemma}   \label{classical}
If $f$ satisfies Assumption \ref{Ass:f_cont} and is $C^3$ then \eqref{PDE1} admits a classical solution $\varphi$.
\end{lemma}
\begin{proof}
By the representation \eqref{est_der} we have
$$\varphi^m_{x^i}(t,a)= e^{-\alpha_i(T-t)}E^Q_{t,a}\left[f_{x^i}(A_T)\right]\leq CE^0_{t,a}\left[\|S_T\|^q\right]\leq C\|s\|^q$$
where the constant is independent of $m$, since this parameter only modifies through $\delta$ the dynamics of $X$, and by the growth assumptions on $f$.\\
For $M>0$ arbitrarily large we can find $D>0$ such that $\gamma\|\beta'\varphi^m_x\|\leq D$ if $\|s\|\leq M$, uniformly in $m$. Therefore if $m\geq D$ then
\begin{equation}\sup_{\delta\in\mathcal B^m(\mathbb R^d)}\left\{-(\beta'\varphi^m_x)\delta-\frac{1}{2\gamma}\|\delta\|^2\right\}=\sup_{\delta\in\mathbb R^d}\left\{-(\beta'\varphi^m_x)\delta-\frac{1}{2\gamma}\|\delta\|^2\right\},\label{sup=sup}\end{equation}
for $\|s\|\leq M$. Since $M$ is arbitrary, this implies that \eqref{PDE1} admits a classical solution on the whole domain $[0,T]\times\mathbb R^n_+\times \mathbb R^d$.
\qed\end{proof}
We can finally prove the part (i) in Theorem \ref{main} for a continuous payoff $f$ satisfying Assumption \ref{Ass:f_cont}.\bigskip

\begin{proof}[ of Theorem \ref{main}\,(i) under Assumption \ref{Ass:f_cont}]
We approximate the payoff by a sequence of $C^3$ functions $f^l$ satisfying Assumption \ref{Ass:f_cont} and converging pointwise to $f$. We assume $f^l$ to be bounded and with bounded derivatives for each $l$.
When a smooth $f^l$ is used as terminal condition by Lemma \ref{classical} we can define the classical solution $\varphi^{l}$ to PDE \eqref{PDE1} as a limit of a sequence $\varphi^{m,l}$ when $m\to\infty$. By Lemma \ref{der_repr} for each $m$ we have
\begin{equation}
|\varphi^{m,l}_{s^i}(t,a)|+|\varphi^{m,l}_{x^j}(t,a)|\leq  C\|s\|^q\wedge l \label{bound_phi}
\end{equation}
where
$$dX^i_t=(b^i (t)-\alpha_i X^i_t+\beta_{i\bcdot}\hat\delta^m _t) dt+\beta_{i\bcdot} dW^{X,Q}_t \quad i=1,\ldots,d,$$
and $\hat \delta^m$ is the maximizer in LHS of (\ref{sup=sup}).
Here $C$ is independent of $m$ (because of the uniformity property in the nontraded assets as in Assumption \ref{Ass:f_cont}) and of $l$ (because of continuity). Moreover, bounding with $l$ can always be done since the derivatives of $f^l$ are bounded for each $l$.
Remark therefore that, being $\gamma \beta' \varphi_x ^l (t,a)$ the maximizer in the RHS of (\ref{sup=sup}), one necessarily has $\|\hat \delta^m _t \| \leq \| \gamma \beta' \varphi_x ^l (t,A_t)\|$. Therefore $\hat\delta^m _t= -\gamma \beta'\varphi^l_x(t,A_t)$ when $m$ is big enough and therefore
\begin{equation}\label{phi-l-exp}\varphi^{l}_{x^j}(t,a)= e^{-\alpha_j(T-t)}E^Q_{t,a}\left[f^l_{x^j}(A_T)\right]\leq  C\|s\|^q\end{equation}
and similarly for $\varphi^{l}_{s^i}$, where
\begin{equation}	\label{dyn_nontraded}
dX^j_t=(b^j (t)-\alpha_j X^j_t-\gamma\beta_j\beta'\varphi^l_x(t,A_t)) dt+\beta_jdW^{X,Q}_t \quad j=1,\ldots,d.
\end{equation}
For fixed $m$ we recall the Zhang representation in \cite{zhang.05}, Theorem 3.2 (as in \eqref{zhang})
\begin{equation*}
\varphi_{a}^{m,l}(t,a)= E^0_{t,a}\left[f^l(A_T)N_T-\int_t^T h^m(Z^{X,m,l}_r)N_r dr\right]
\end{equation*}
where
$$Y^{m,l}_t=f^l-\int_t^T h^m(Z^{X,m,l}_r)dr-\int_t^T Z^{m,l}_r dW^0_r .$$
Hence
\begin{equation*}
\varphi_{a}^{l}(t,a)= E^0_{t,a}\left[f^l(A_T)N_T-\frac{\gamma}{2}\int_t^T\|Z^{X,l}_r\|^2N_r dr\right]
\end{equation*}
by dominated convergence and the previous estimates (\ref{bound_phi}) applied to
$$Z^{X,m,l}_t=\sigma (S_t)' \varphi_s ^{m,l}(t,A_t),$$
and the fact that $Z^{X,m,l}\to Z^{X,l}$ in $\mathbb H^{q'} (\mathbb R^d)$ for all $q' >0$ as $m\to \infty$ using classical results on quadratic BSDEs in \cite{Kob.00} (since we can assume without loss of generality that $f^l$ is bounded for fixed $l$), where
\begin{equation}	\label{BSDE_l}
Y^{l}_t=f^l-\int_t^T  \frac{\gamma}{2}\|Z^{X,l}_r\|^2dr-\int_t^T Z^{l}_r dW^0_r.
\end{equation}
Now by using an argument like in Lemma \ref{existence_BSDE} we get that $Y^{l}\to Y$ as $l\to \infty$ where
\begin{equation}    \label{conv_l}
Y_t=f-\int_t^T \frac{\gamma}{2}\|Z^{X}_r\|^2dr-\int_t^T Z_r dW^0_r
\end{equation}
and also $Z^{l}\to Z$ in $\mathbb H^{q'} (\mathbb R^d)$ for all $q'> 0$. By the definition of the process $N$ in (\ref{def-N}), we obtain that $E^0_t[\|N_T\|^p]\leq C(T-t)^{-p/2}$ for any $p\geq 1$ and some constant $C>0$. Therefore again by dominated convergence
\begin{equation*}
\varphi_{a}^{l}(t,a)\to g(t,a) := E^0_{t,a}\left[f(A_T)N_T-\frac{\gamma}{2}\int_t^T\|Z^{X}_r\|^2N_r dr\right].
\end{equation*}
Similarly as in the last part of our Lemma \ref{visc_solution}, using Lemma 6.2 in \cite{fleming.06} we deduce that $\varphi^l$ converges, uniformly on compact sets of $[0,T]\times\mathbb R^n_+\times \mathbb R^d$, to $\varphi$, viscosity solution to \eqref{PDE1}, which is also continuous.\\
We will now show that $g$ is continuous and that $g=\varphi_{a}$.
To do so we can adapt the last part of Zhang's proof of Theorem 3.2 in \cite{zhang.05}, we give all the details for reader's convenience.
For all $\varepsilon>0$ we can choose an open set $O_\varepsilon$ with Lebesgue measure smaller than $\varepsilon$ and a continuous function $f^\varepsilon$ such that $f^\varepsilon=f$ outside $O_\varepsilon$.
Denote
$$g_\varepsilon(t,a):= E^0_{t,a}\left[f^\varepsilon(A_T)N_T-\frac{\gamma}{2}\int_t^T\|Z^{X}_r\|^2N_r dr\right]$$
(where $Z$ is solution to the limit BSDE  \eqref{conv_l}, with $f$ and not $f^\varepsilon$ as terminal condition). Denoting $g^i$ and $g^i _\varepsilon$ the $i$-th component of, respectively, $g$ and $g_\varepsilon$ we get
\begin{equation*}
\begin{split}
|g^i_\varepsilon-g^i|(t,a)&= |E^0_{t,a}\left[(f^\varepsilon(A_T)-f(A_T))N^{i}_T\right]|\leq E^0_{t,a}\left[|f^\varepsilon(A_T)-f(A_T)||N^{i}_T|;X_T\in O_\varepsilon\right]\\
&\leq E^0_{t,a}\left[|f^\varepsilon(A_T)+f(A_T)||N^{i}_T|;X_T\in O_\varepsilon\right]\leq C(t,a)\sqrt{\varepsilon}.
\end{split}
\end{equation*}
for some constant $C(t,a)$. Now taking a sequence $(t_{\kappa},A_{\kappa})$ tending to $(t,a)$ we have
\begin{equation*}
\begin{split}
|g^i&(t_{\kappa},A_{\kappa})-g^i(t,a)|\\
&\leq |g^i(t_{\kappa},A_{\kappa})-g^i_\varepsilon(t_{\kappa},A_{\kappa})|+|g^i_\varepsilon(t_{\kappa},A_{\kappa})-g^i_\varepsilon(t,a)|+|g^i_\varepsilon(t_{\kappa},A_{\kappa})-g^i(t,a)|\\
&\leq [C(t,a)+C(t_{\kappa},A_{\kappa})]\sqrt{\varepsilon}+|g^i_\varepsilon(t_{\kappa},A_{\kappa})-g^i_\varepsilon(t,a)|.
\end{split}
\end{equation*}
Since $g^i_\varepsilon$ is continuous and $\varepsilon$ is arbitrary we deduce that $g^i$ is continuous as well. 
Now for any $(t,\tilde a)\in [0,T]\times\mathbb R^n_+\times \mathbb R^d$ we have
$$\varphi^l(t,\tilde a)=\varphi^l(t,I^i\tilde a)+\int_0^{\tilde a^i} \varphi^l_{a^i}(t,I^i\tilde a+e_iy)dy$$
where we denoted $I^i$ the $\mathbb R^{n+d}$-identity matrix whose $i$-th diagonal entry is zero, and $e_i$ is the canonical basis vector in $\mathbb R^{n+d}$. By dominated convergence (using \eqref{bound_phi}) we deduce
$$\varphi(t,\tilde a)=\varphi(t,I^i\tilde a)+\int_0^{\tilde a^i} g^i(t,I^i\tilde a+e_iy)dy,$$
 implying that $g=\varphi_{a}$.
\qed\end{proof}
\begin{remark} In the representation \eqref{dyn_nontraded} it would be tempting to pass from measure $Q$ (coming from the stochastic control representation) to the MMM $Q^0$ by identifying $$dW^{X,0}_t=dW^{X,Q}_t-\gamma\beta'\varphi^l_x(t,A_t) dt.$$ This may however not be possible in general due to the growth properties of $\varphi^l_x$ and the fact that geometric Brownian motion does not have exponential moments.
\\We will perform a similar change of measure in the next section under more restrictive assumptions on the derivatives of the payoff function $f$.\end{remark}

\subsubsection{Discontinuous payoffs}

In this part of the paper, we show that the continuity of the payoff $f$ can be partially removed. The price to pay for that is imposing stronger conditions on its derivatives as in Assumption \ref{Ass:f_noncont}.
\\The idea hat lies at the heart of the proof that follows is showing that, when we approximate our discontinuous payoff $f$ with a smooth sequence $f^l$, the derivatives of the price $\varphi^l$ will not explode when $l\to \infty$ for $t<T$. This is easily seen if we take, for example, the digital payoff $f(x)=\mathbf 1_{[0,\infty)}(x)$ which does not depend on the traded assets. Setting $\alpha=0$ in the dynamics (\ref{X_dynamics}) we have
$$\varphi^l_x(t,x)=E^Q[f^l(X_T)]$$
with
$$dX_t=-\gamma\varphi^l_x(t,X_t)dt+\beta dW^{X,Q}_t,$$
and $\varphi^l_x(T-t,x)\to g(t,x)$, where $g$ solves the Burgers' equation
$$g_t+\gamma g_x g=\frac{1}{2}\beta^2g_{xx}$$
which has the solution
\begin{equation*}
g(t,x)=\frac{ \beta e^{-\frac{x^2}{2\beta^2 t}}(1-e^{-\frac{\gamma}{\beta^2}})}{\gamma\sqrt{2\pi t}\left[(e^{-\frac{\gamma}{\beta^2}}-1)\Phi\left(\frac{x}{\beta\sqrt{t}}\right)+1\right]}
\end{equation*}
In particular we clearly have $g(t,x)\leq\frac{C}{\sqrt{t}}$, where $C=\frac{ \beta}{\gamma\sqrt{2\pi }}(e^{\frac{\gamma}{\beta^2}}-1)$. Unfortunately the Burgers-type equation that results by adding traded assets does not seem to have an explicit solution, therefore we will need to employ a different method to get a similar estimate. Here is the proof of our main result concerning discontinuous payoffs.\medskip

\begin{proof}[ of Theorem \ref{main}\,(i) under Assumption \ref{Ass:f_noncont}]
Take again a sequence $f^l$ of approximating smooth functions as in the proof of Lemma \ref{solution_PDE}. Each function $f^l$ of the sequence satisfies the assumptions of Lemma \ref{der_repr}, so that the representation formula therein applies and we have that
\begin{equation}    \label{est_x}
|\varphi_{x^i}^{l}(t,a)|\leq C^l (1+\|s\|^q),
\end{equation}
with the constant $C^l$ depending on $l$. Remark that this is not the same constant appearing in the characterization of uniform growth with respect to $x$: since we are dealing with discontinuous payoffs, the derivatives of the approximating functions $f^l$ may well explode close to the discontinuities for large $l$. We will have
\begin{equation}    \label{est_fl}
|f_{x^i}^{l}(a)|\leq C^l(x) (1+\|s\|^q) \quad i=1,\ldots,d,
\end{equation}
where $C^l(x)$ is a function which stays bounded on compact sets which do not include discontinuity points, but that may explode at these points for large $l$. In order to see this, we can explicitly write the mollified sequence $f^l$ as
$$f^l(s,x)=\int_{\mathbb R^d} f(s,x+y)\psi^l(y)dy=\int_{\mathbb R^d} f(s,z)\psi^l(z-x)dz$$
where
$$\tilde \psi^l(x)=K\exp\left(\frac{-1}{1-\|x\|^2}\right)\mathbf 1_{\{\|x\|\leq 1\}} , \quad \psi^l(x)=l\tilde\psi^l\left(l x\right)$$
Recall that $\psi^l$ is a mollifier with support on $\mathcal B_d(1/l)$.
If $\|x-I\|>1/l$, where $I$ is the discontinuity point closest to $x$, then
$$f^l_{x^i}(s,x)=\int_{\mathbb R^d} f_{x^i}(s,x+y)\psi^l(y)dy$$
and so $|f^l_{x^i}(s,x)|\leq C(1+\|s\|^q)$.
For $\|x-I\|\leq1/l$ we use the representation (recall that $f(s,\cdot)$ is bounded for fixed $s$)
$$f^l_{x^i}(s,x)=-\int_{\mathbb R^d} f(s,z)\psi^l_{x^i}(z-x)dz$$
which yields
$$|f^l_{x^i}(s,x)|\leq C(1+\|s\|^q)\int_{\mathbb R^d} |\psi^l_{x^i}(z-x)|dz\leq Cl(1+\|s\|^q)$$
since $f$ has uniform polynomial growth in $s$. Therefore
$$|f^l_{x^i}(s,x)|\leq C^l(x)(1+\|s\|^q)$$
where
$$C^l(x)=Cl\mathbf 1_{\{\|x-I\|\leq1/l\}}.$$
Also by Lemma \ref{der_repr} and Assumption \ref{Ass:f_noncont} (iii) we have
\begin{equation}    \label{est_s}
|\varphi^l_{s^i}(t,a)|\leq C\frac{1}{s^i}
\end{equation}
for $s^i$ big enough (since discontinuities can only occur in the $x$-variables) and for some constant $C>0$ independent of $l$ and $x$. If we consider the pricing BSDE \eqref{BSDE_l} associated with $f^l$ we can identify
$Z^{X,l}_t=\beta'\varphi^l_x(t,A_t)$ and $Z^{S,l}_t=\sigma(S_t)'\varphi^l_s(t,A_t)$.
By estimate \eqref{est_s} we deduce that $Z^{S,l}$ is bounded for each $l$ (with a bound independent on $l$). By estimate \eqref{est_x} we can assume $Z^{X,l}$ to be bounded for each $l$ (by possibly bounding the growth in the traded assets with $l$), which allows us to perform a probability measure change to get
\begin{equation}    \label{probchange}
\begin{split}
|\varphi_{x^i}^{l}(t,a)|&=^{(i)}\left|E^0_{t,a}\left[\frac{\mathcal E_T}{\mathcal E_t}(-\gamma Z^{X,l}\cdot W^X) e^{-\alpha_i(T-t)} f^l_{x^i}(A_T)\right]\right|\\
&\leq^{(ii)} CE^0_{t,a}\left[e^{\gamma(Y^l_t-f^l+\int_t^TZ^{S,l}_r dW^{S,0}_r)} |f^l_{x^i}(A_T)|\right]\\
&\leq^{(iii)} Ce^{\gamma Y^l_t}E^0_{t,a}\left[e^{\gamma \int_t^TZ^{S,l}_r W^{S,0}_r} |f^l_{x^i}(A_T)|\right]\\
&\leq^{(iv)} C e^{\gamma Y^l_t}E^0_{t,a}\left[\frac{\mathcal E_T}{\mathcal E_t}(\gamma Z^{S,l}\cdot W^{S,0}) |f^l_{x^i}(A_T)|\right]\\
&=^{(v)} C e^{\gamma Y^l_t}E^{\bar Q}_{t,a}\left[|f^l_{x^i}(A_T)|\right]\leq^{(vi)} C\|s\|^q e^{\gamma Y^l_t}E^{0}_{t,x}\left[C^l(X_T)\right]\\
&\leq^{(vii)} \frac{C\|s\|^q\|x\|^{q'}}{\sqrt{T-t}} e^{\gamma Y^l_t}\leq^{(viii)}  \frac{C\|s\|^q\|x\|^{q'}}{\sqrt{T-t}}  e^{\gamma C(1+\|s\|^q)}\end{split}
\end{equation}
where the constant $C$ changes from line to line and the inequalities above can be justified as follows:
\begin{enumerate}
\item is due to the second equality in (\ref{est_der}) applied to the sequence $\varphi^l (t,a)$, which has bounded derivatives.
\item comes from the pricing BSDE \eqref{BSDE_l} under the MMM $Q^0$, which implies
$$\frac{\mathcal E_T}{\mathcal E_t}(-\gamma Z^{X,l}\cdot W^X)=e^{-\gamma\left( \int_t^TZ^{X,l}_rdW^X_r+\frac{\gamma}{2} \int_t^T\|Z^{X,l}_r\|^2dr\right)}=e^{\gamma(Y^l_t-f^l+\int_t^TZ^{S,l}_r dW^{S,0}_r)} .$$
\item is a consequence of boundedness from below of $f$.
\item is derived from boundedness of $Z^{S,l}$, uniformly in $l$ (so that $C$ does not depend on $l$).
\item is obtained by applying the measure change $\frac{d\bar Q}{d Q^0}=\mathcal E_T(\gamma Z^{S,l}\cdot W^{S,0})$.
\item the inequality comes from Assumption \ref{Ass:f_noncont}\,(ii) and the fact that the drift changes   induced by the measure change $\frac{d\bar Q}{d Q^0}$ are bounded and only pertain the tradable assets. In particular the dynamics of $S^i$ under $\bar Q$ can be controlled by noticing
$$S^i_T=S^i_te^{\gamma \sigma_{i\bcdot}\int_t^TZ^{S,l}_udu-\frac{\|\sigma_{i\bcdot}\|^2}{2}(T-t)+\sigma_{i\bcdot}(W^{S,\bar Q}_T-W^{S,\bar Q}_t)}\leq CS^i_te^{-\frac{\|\sigma_{i\bcdot}\|^2}{2}(T-t)+\sigma_{i\bcdot}(W^{S,\bar Q}_T-W^{S,\bar Q}_t)}.$$ The inequality above is due to the fact that, by Assumption \ref{Ass:f_noncont}\,(ii), there exist a threshold $M>0$ such that $|\varphi_{s^i} ^l (t,A_t)| \leq C/S_t ^i$ when $|S_t ^i | \geq M$, otherwise it is bounded. Thus one obtains
\begin{equation} |\gamma Z_t^{S,l}| = |\gamma \sigma (S_t)' \varphi_s ^l (t,A_t)| \leq C \gamma \left\| \sigma (S_t)\frac{1}{S_t} \right\| ,\label{est-Z-disc}\end{equation}
where one can easily check that the last term on the RHS is constant.

\item is the derived from the definition of $C^l$ and using the density of $X_T$ (i.e. the multivariate Gaussian). In fact, taking for simplicity just one discontinuity point at zero we immediately see that
$$E^{0}_{t,x}\left[C^l(X_T)\right]=Cl P_{t,x}(\|X_T\|\leq 1/l)\leq Cl \frac{1}{l}\frac{1}{\mbox{det}(\textrm{Var}_{t,x}(X_T))^{1/2}}\leq \frac{C}{\sqrt{T-t}}$$
with the obvious notations for conditional variance and probability.
\item Since $f$ has uniform polynomial growth in $s$, the same holds for $f^l$ (uniformly in $l$). Therefore $Y^l_t\leq E^0_{t,a}[f^l(A_T)]\leq C+CE^0_{t,s}[\|S_T\|^q]\leq C(1+\|s\|^q)$.
\end{enumerate}
Using the previous estimate (\ref{probchange}), we can apply the usual stability properties (Lemma 6.2 in \cite{fleming.06}) to get $Y_t=\lim_l\varphi^l(t,A_t)=\varphi(t,A_t)$, where $\varphi$ is a viscosity solution of \eqref{PDE1}.
\\We now would like to prove that $\varphi$ has continuous first derivatives in all space variables. Since $Z^{X,l}$ is locally bounded uniformly in $l$ by \eqref{probchange} we can use Lemma \ref{unif_int} componentwise (together with Lemma \ref{L2}) to get the uniform integrability property allowing us to use dominated convergence and obtain
\begin{equation*}
\varphi_{a}^{l}(t,a)\to g(t,a):= E^0_{t,a}\left[f(A_T)N_T-\frac{\gamma}{2}\int_t^T\|Z^{X}_r\|^2N_r dr\right] .
\end{equation*}
To conclude it suffices to show that $g$ is continuous and that $g=\varphi_{a}$. This can be done by exactly the same arguments that we used at the end of the proof of Theorem \ref{main}\,(i) under Assumption \ref{Ass:f_cont}. For this reason, we omit this part of the proof.
\qed\end{proof}

\begin{remark}
Had we supposed directly the multiplicative form $f(s,x)=g(x)h(s)$ with a bounded $g$ then we could have allowed for a countable (and not simply finite) number of discontinuities in $g$. This is true by remarking that in \eqref{probchange} we could have used Theorem 3.2 in \cite{zhang.05}, by considering the function $u^l (t,x)=E^{0}_{t,x}[g^l(X_T)]$ (corresponding to the trivial linear BSDE arising from the martingale representation theorem) and the estimates on its derivative $u^l _x(t,x)=E^{0}_{t,x}[g^l_x(X_T)]$.
\end{remark}

\begin{remark}
  Here we focused on the case of discontinuities only taking place in the $x$-variables, as it turns out to be the most useful case in the applications (See Section \ref{Sec:electricity}). The arguments in the previous proof (in particular estimate \eqref{probchange}) could, however, be easily adapted to the case where discontinuities take place only in the $s$ variables, provided the payoff has polynomial growth in $x$, uniformly in $s$.
\end{remark}

\subsubsection{The optimal hedging strategy}
The previous results (stating the differentiability of UIP) allows us to represent $Z^S$ in terms of the derivatives of the solution of a PDE. Indeed, when $f$ is bounded, the optimal strategy can be immediately recovered by $\Delta_t=-\sigma^{-1}Z^S_t$, using Lemma \ref{UIP_BSDE}. The next result gives a slight generalization to the case where $f$ has polynomial growth in the traded assets.\medskip

\begin{proof}[ of Theorem \ref{main}\,(ii)] 
Approximate $f$ as in Lemma \ref{solution_PDE} with a sequence $f^l$, where each of its element can always be taken to be bounded. By Lemma \ref{UIP_BSDE}, the corresponding optimal strategies with the claims $f^l$ are given by $\hat\pi^{l}_t=-\sigma^{-1}\sigma(S_t)\varphi^l_s(t,A_t)+\frac{1}{\gamma}\sigma^{-2}\mu$ and the value functions are
\begin{equation*}
u^l(t,v,a)=\sup_{\pi}E_{t,a}\left[-e^{-\gamma\left(V^v_T(\pi)+ f^l\right)}\right]=E_{t,a}\left[-e^{-\gamma\left(V^v_T(\hat\pi^{l})+ f^l\right)}\right].
\end{equation*}
By the growth assumptions in $s$ (uniform in $x$) we deduce that the assumptions of Lemma \ref{convergence} are satisfied and therefore
\begin{equation} \label{ul_to_u}
u^l\to u
\end{equation}
for all $(t,v,a)\in [0,T]\times \mathbb R\times \mathbb R^n_+\times \mathbb R^d$, where
$$u(t,v,a)=E_{t,a}\left[-e^{-\gamma\left(V^v_T(\hat\pi)+ f\right)}\right]$$
for some optimal $\hat\pi$. We would like to identify $\hat\pi$ with $\tilde\pi_t :=-\sigma^{-1}\sigma(S_t)\varphi_s(t,A_t)+\frac{1}{\gamma}\sigma^{-2}\mu$.
An application of the reverse Fatou's Lemma gives
\begin{equation} \label{Fatou1}
\limsup_l E_{t,a}\left[-e^{-\gamma\left(V^v_T(\hat\pi^{l})+ f^l\right)}\right]\leq  E_{t,a}\left[\lim_l -e^{-\gamma\left(V^v_T(\hat\pi^{l})+ f^l\right)}\right],
\end{equation}
where the limit on the LHS is meant to be in probability. To show that this limit exists, remark first that $\hat\pi^{l}\to\tilde\pi$ in $\mathbb H^{2} (\mathbb R^n)$, which implies that $V^v_T(\hat\pi^{l})$ converges to $V^v_T(\tilde\pi)$ in $L^2(\Omega,P)$, hence in probability. In the same way, $f^l\to f$ in probability. By using \eqref{ul_to_u} and the continuity of the exponential function, \eqref{Fatou1} becomes
\begin{equation*}
E_{t,a}\left[-e^{-\gamma\left(V^v_T(\hat\pi)+ f\right)}\right]\leq  E_{t,a}\left[ -e^{-\gamma\left(V^v_T(\tilde\pi)+ f\right)}\right],
\end{equation*}
which implies that $\tilde \pi$ is optimal. Indeed, one can show that $\tilde \pi$ belongs to $\mathcal H_M$ using the uniform estimate (\ref{bound_phi}) in the continuous payoff case or the estimate (\ref{est-Z-disc}) in the discontinuous case together with the fact that $S$ has moments of all positive orders.
\qed\end{proof}

\subsection{Asymptotic expansions}
In this subsection we turn to the problem of computing effectively the UIP and the corresponding optimal hedging strategy for a given contingent claim.
It is well-known that solving PDE \eqref{PDE1} numerically can be impractical for time reasons when the number of assets is large. It is therefore useful to derive some asymptotic expansions which allow to approximate the price and the hedging strategy when the risk aversion parameter $\gamma$ is small. The formulas are given in terms of the no-arbitrage price and strategy, which can usually be computed in a much simpler way either explicitly or by numerical integration or by Monte Carlo methods.\\
Consider a contingent claim with payoff $f(A_T)$ integrable under the MMM $Q^0$, whose no-arbitrage price under $Q^0$ is denoted by $p^0(t,a)=E^0_{t,a}[f(A_T)]$. 
Now define
$$\zeta(t,a):=E^0_{t,a}\left[\int_t^T \|\beta' p^0_{x}\|^2(s,A_s)ds\right].$$
The next result is due to a recent preprint by Monoyios (\cite{monoyios.12}).
\begin{lemma}\label{asymp}
Under Assumption \ref{Ass:f_cont} or \ref{Ass:f_noncont} for the contingent claim $f(A_T)$, the following asymptotic expansion holds:
  \begin{equation}	\label{asymptotic}
  \varphi(t,a)=p^0(t,a)-\frac{\gamma}{2}\zeta(t,a)+O(\gamma^2).
  \end{equation}
\end{lemma}
\begin{proof}
This is a reformulation of  \cite{monoyios.12}, Theorem 5.3. It is enough to remark that our growth assumptions on $f$ ensure that it is in $L^2(Q)$ for any $Q\in\mathcal M_E$.
\qed\end{proof}
The next result provides asymptotic expansions for the derivatives of the price, and therefore of the optimal hedging strategy.

\begin{lemma}
  Suppose Assumption \ref{Ass:f_cont} holds, and moreover that $f_x$ is bounded. Then the following asymptotic expansions hold
\begin{eqnarray*} \varphi_{x^i}(t,a) &=& e^{-\alpha_i (T-t)} E^0_{t,a}\left[f_{x^i}(A_T)\right]-\gamma e^{-\alpha_i (T-t)}  E^0_{t,a}\left[f_{x^i}(A_T) \int_t^T \beta' \varphi^0_{x} (u,A_u) d W^X_u\right]+O(\gamma^2) \label{asydev_x} \\
\varphi_{s^i}(t,a) &=& E^0_{t,a}\left[\frac{S^i_T}{S^i_t}f_{s^i}(A_T)\right]-\gamma E^0_{t,a}\left[\frac{S^i_T}{S^i_t}f_{s^i}(A_T) \int_t^T \beta' \varphi^0_{x} (u,A_u) d W^X_u\right]+O(\gamma^2), \label{asydev_s}
\end{eqnarray*}
where $\varphi^0_{x^i}(t,a)= e^{-\alpha_i (T-t)} E^0_{t,a}\left[f_{x^i}(A_T)\right]$.
\end{lemma}

\begin{proof}
In the rest of the proof for simplifying the notation, we will prove the expansions for $\alpha_i =0$ and only at time $t=0$, otherwise the same arguments (conditionally to $\mathcal F_t$) apply and get the result for any $t$. By considering as usual a sequence of approximating functions we get from equality (\ref{phi-l-exp}) and a simple application of Girsanov's theorem that
$$\varphi_{x^i}^{l}(0,a)=E^0 \left[\mathcal E_T(-\gamma \beta'\varphi^{l}_{x}\cdot W^X) f^l_{x^i}(A_T)\right]$$
which is bounded, uniformly in $l$. By taking $l\to\infty$ we get
$$\varphi_{x^i}(0,a)=E^0 \left[\mathcal E_T(-\gamma \beta'\varphi_{x}\cdot W^X) f_{x^i}(A_T)\right],$$
which is also bounded. Now we write $\varphi^\gamma$ to emphasize dependence on $\gamma$. So we have
$$\frac{\varphi^\gamma_{x^i}-\varphi^0_{x^i}}{\gamma}(0,a)=E^0 \left[\frac{\mathcal E_T(-\gamma \beta'\varphi^\gamma_{x}\cdot W^X)-1}{\gamma} f_{x^i}(A_T)\right] .$$
Moreover, we will denote the process $\varphi_x ^\gamma (t,A_t)$ by $\varphi_x ^\gamma$ with a slight abuse of notation. Remark that, defining $M^\gamma$ as the unique solution to $dM^\gamma_t=-\gamma M^\gamma_t\beta'\varphi^\gamma_{x}(t,A_t)dW_t ^X$ with initial condition $M_0 ^\gamma=1$, we have
\begin{equation*}
\begin{split}
E^0&\left[\left(\frac{\mathcal E_T(-\gamma \beta'\varphi^\gamma_{x}\cdot W^X)-1}{\gamma}+\int_0^T \beta'\varphi^0_{x}d W^X_s\right)^2\right]=E^0\left[\left(\int_0^T (\beta'\varphi^0_{x}-M^\gamma_s\beta'\varphi^\gamma_{x})d W^X_s\right)^2\right]\\
&=E^0\left[\int_0^T \|\beta'\varphi^0_{x}-M^\gamma_s\beta'\varphi^\gamma_{x}\|^2ds\right]\\
&\leq 2 E^0\left[\int_0^T \|\beta'\varphi^0_{x}-\beta'\varphi^\gamma_{x}\|^2ds\right]+2 E^0\left[\int_0^T \|\beta'\varphi^\gamma_{x}\|^2(1-M^\gamma_s)^2ds\right]\\
&\leq C E^0\left[\int_0^T(1-M^\gamma_s)^2ds\right], \\
\end{split}
\end{equation*}
where the second equality is due to It\^o's isometry, since the integrand therein belongs to $\mathbb H^2 (\mathbb R^d)$.
Since $f_{x^i}$ is bounded by assumption, $\varphi^\gamma_{x}$ is also bounded and this implies that $E^0[\int_0 ^T (1-M_s ^\gamma)^2 ds ]$ tends to zero as $\gamma \to 0$ by dominated convergence. Thus
$$\frac{\mathcal E_T(-\gamma \beta'\varphi^\gamma_{x}\cdot W^X)-1}{\gamma}\to -\int_0^T \beta'\varphi^0_{x}d W^X_t$$
in $L^2$ as $\gamma\to 0$, and therefore
$$\left.\frac{\partial}{\partial \gamma}\varphi^\gamma_{x^i}\right|_{\gamma=0}=\lim_{\gamma\to 0}\frac{\varphi^\gamma_{x^i}-\varphi^0_{x^i}}{\gamma}=-E^0 \left[f_{x^i}(A_T) \int_0^T \beta'\varphi^0_{x}d W^X_s\right].$$
The proof for $\varphi_{s^i}$ is analogous.
\qed\end{proof}
We conclude this section with a lower bound on the utility indifference price of $f$.
\begin{lemma}
Under Assumptions \ref{Ass:f_cont} or \ref{Ass:f_noncont} the following bound on the price holds:
$$\varphi(t,a)\geq-\frac{1}{\gamma}\log E^0_{t,a}\left[e^{-\gamma f(A_T)}\right] .$$
\end{lemma}
\begin{proof}
Define
$$h(t,a)=E^0_{t,a}\left[e^{-\gamma f(A_T)}\right]$$
which solves
\begin{equation*}
\begin{cases}
\mathcal L h=0\\
h(T,a)=e^{-\gamma f(a)}
\end{cases}
\end{equation*}
in the classical sense (assuming $f$ to be smooth). Now set $g=-\frac{1}{\gamma}\log h$, so that $g$ solves
\begin{equation*}
\begin{cases}
\mathcal L g -\frac{\gamma}{2}\|\sigma(S)'g_s\|^2-\frac{\gamma}{2}\|\beta'g_x\|^2=0\\
g(T,a)=f(a).
\end{cases}
\end{equation*}
By the comparison theorem for PDEs we have that $g(t,a)\leq\varphi(t,a)$. By our approximation arguments the same bound holds true when $f$ is not smooth. 
\qed\end{proof}


\section{Application to electricity markets}		\label{Sec:electricity}
Our framework can be particularly useful to evaluate derivatives in situations where the underlying asset prices are determined by the interplay between several factors, but only some of these can be actually traded on a financial market (while the others may be of a totally different nature, for example macroeconomic or even behavioral factors). This is the case in particular for structural models of electricity prices, where the relevant components that influence the price are typically both tradable (like fuels) and non tradable (like market demand or production capacities)\footnote{We refer the reader to \cite{Carmona.12} for a comprehensive survey of structural models.}.
\\The seminal contribution in the direction of structural electricity models has been the Barlow's model (\cite{Barlow.02}), which describes the electricity spot price as a function of a one-dimensional diffusion representing the evolution of market demand. Since there is only a non tradable asset in his framework, utility indifference valuation here reduces to the computation of the certainty equivalent (see Remark \ref{rmk_certainty}), at least when prices are bounded (an assumption which is suggested by Barlow himself and which reflects the reality of electricity markets, where prices are usually capped).
Similar considerations hold for the models in \cite{skantze.00} or \cite{Cartea.08}, where an exponential function is used and an additional non tradable factor is added describing maximal capacity.\\
Building on this literature, several authors have proposed more developed structural models with the aim of capturing the contribution of other assets, notably the (marginal) fuels employed in electricity generation along with their production capacities. Since fuels are commodities which are typically traded on financial markets, their introduction fully justifies the employment of pricing techniques that allow for some kind of partial hedging (such as local risk minimization or, in our case, utility indifference pricing). For example, in \cite{pirrong.08} the authors describe the spot price as the product of two components accounting for a traded and a non traded asset (following, respectively, a geometric Brownian motion and an Ornstein-Uhlenbeck process as in our framework). Multi-asset models have then followed, with the aim of considering the whole stack of available fuels, which typically present different levels of correlation with the spot price depending on their available capacities and market demand. They enter in our framework, possibly with some minor adaptations.\medskip \\
In this paper we will focus especially on the model introduced in \cite{Aid.10}, where the authors directly model the spreads between fuels as geometric Brownian motions, hence the tradable assets of our model $S^i_t$ can be interpreted in this case as those fuel spreads by using the relation
$$S^i_t=h_i K^i_t-h_{i-1}K^{i-1}_t,$$
where $K_t ^i$ is the price at time $t$ of $i$-th fuel and the $h_i$'s are heat rates associated to each fuel. The model also includes fuel capacities $C^i_t$ and a process $D_t$ describing the demand for electricity, which make for $d=n+1$ nontradable assets. In \cite{Aid.10}, the dynamics postulated for tradable and nontradable assets perfectly fit into our setting, since the spread between two fuels follows a multidimensional Black-Scholes model while the non tradable ones follow Ornstein-Uhlenbeck processes with non zero mean-reversion and a seasonality component that can be embedded in the function $b(t)$ as in (\ref{X_dynamics}). More precisely, we have
\begin{eqnarray}	\label{dynamics2}
\frac{dS^i_t}{S^i_t}&=&\mu_i dt+\sigma_{i} dW^{S,i}_t,\quad i=1,\ldots ,n\\
dC^j_t&=&(b_{C^j} (t) -\alpha_{C^j} C^j_t) dt+\beta_{C^j} dW^{C^j}_t 	\quad j=1,\ldots,n\\
dD_t &=& (b_D (t) -\alpha_D D_t) dt+\beta_D dW^{D}_t 	,
\end{eqnarray}
where we also supposed that the stochastic components of the assets are independent (compare with equation (4.2) in \cite{Aid.10}), i.e. the Brownian motions $W^{C^j}$ and $W^{D}$ are assumed to be independent. The coefficients $\mu_i, \alpha_{C^j}, \alpha_D$ are arbitrary constants while $\sigma_i, \beta_{C^j}, \beta_D$ are strictly positive real numbers. Moreover, $b_{C^j} (t)$ and $b_D (t)$ are deterministic bounded functions that possibly include the seasonality component of nontraded asset dynamics.\\
One of the main goals of structural models for energy markets (included the one in \cite{Aid.10}) is to have a realistic and tractable  setting where pricing and hedging power derivatives. One of the most important derivatives to price and hedge is the forward contract on electricity, with payoff given by the value at maturity of the electricity spot price, which in \cite{Aid.10} can be written as
\begin{equation}	\label{forward_payoff}
f(a)=f(s,c,y)=g\left(\sum_{i=1}^n c^i-y\right)\displaystyle\sum_{i=1}^n h_i k^i \mathbf 1_{\{y\in I^i\}}=g\left(\sum_{i=1}^n c^i-y\right)\displaystyle\sum_{j\leq i\leq n} s^j\mathbf 1_{\{y\in I^i\}}
\end{equation}
where $g$ is a bounded function with bounded first derivatives, $c^i$ and $y$ stand for fuel capacities and market demand, and we used the fact that $h_i K^i_t=\sum_{j\leq i} S^j_t$. The function $g$ is called \emph{scarcity function}, it has a crucial role for producing spikes in electricity spot prices (see the paper \cite{Aid.10} for further details).\medskip \\
The BSDE approach developed in Section \ref{Sec:UIP_BSDE} can be applied to get the buying UIP of a forward contract written on electricity spot prices. Indeed, for the payoff \eqref{forward_payoff} (as well as for call options on spread) the sufficient conditions established in Lemma \ref{convergence} are easily checked, due to the simple multiplicative structure of the set of equivalent martingale densities implied by the independence between tradable and non tradable assets. On the other hand the payoff \eqref{forward_payoff}, as it is, does not satisfy neither Assumption \ref{Ass:f_cont} or \ref{Ass:f_noncont}, however it can be made to satisfy
\begin{itemize}
\item Assumption \ref{Ass:f_cont} by suitably modifying the scarcity function as in, e.g., \cite{Aid.12}, where the payoff of a forward contract is a Lipschitz continuous functions of all the assets.
\item Assumption \ref{Ass:f_noncont} by bounding the payoff by some constant $M$ (which makes sense since in reality, as already remarked, electricity prices are capped).
\end{itemize}
The same observations hold for the utility indifference pricing of the quite popular spread options, which present a payoff which is either bounded or linearly growing in the electricity price.

\begin{remark} Substantially equivalent considerations hold for the electricity spot price model proposed in \cite{Carmona_sch.12} (equation (6)), which still uses a multiplicative form separating the contributions of traded and non traded assets (in a more involved way than in \cite{Aid.10}, with the drawback of becoming rather messy when more than two assets are considered): bounding the payoff of the forward contract makes it satisfy Assumption \ref{Ass:f_noncont} (remark that it is usually discontinuous in the non traded assets). More generally, as reported in \cite{Carmona.12} (Chapter 5), most of the structural models found in the literature assume lognormal fuel prices, OU-driven demand and an electricity price which is multiplicative in the marginal fuel, which justifies our standing assumptions. Markov switching models like the one described in \cite{Carmona.12}, equation (10), can also be treated in our framework as the structure of the payoff is standard, and additional indicator functions can be added to describe the different regimes (which create discontinuities in the non traded assets).\end{remark}
When the payoff $f$ is linear or concave in the traded assets (as in the case of the forward contract in \cite{Aid.10}) we have the following result.
\begin{lemma}
If $f(s,x)$ is concave in $s$, the same holds for its UIP $\varphi (t,s,x)$.
\end{lemma}
\begin{proof}
By Lemma \ref{solution_PDE} and using an approximating sequence $f^l$, the price is represented as
$$\varphi^l(t,s,x)= E^Q_{t,a}\left[\frac{1}{2\gamma}\int_t^T \|\hat\delta_r\|^2 dr+f^l(S_T,X_T)\right]$$
and therefore, setting $\tilde a = (\tilde s , x)$, we have
\begin{equation*}
\begin{split}
  \varphi^l&(t,\lambda s+(1-\lambda)\tilde s,x)\\
  &\geq\lambda E^Q_{t,a}\left[\frac{1}{2\gamma}\int_t^T \|\hat \delta_r\|^2 dr+f^l(S_T,X_T)\right]+(1-\lambda) E^Q_{t,\tilde a}\left[\frac{1}{2\gamma}\int_t^T \|\hat \delta_r\|^2 dr+f^l(S_T,X_T)\right]\\
  &\geq \lambda \inf_{\delta}E^Q_{t,a}\left[\frac{1}{2\gamma}\int_t^T \|\delta_r\|^2 dr+f^l(S_T,X_T)\right]\\
  &\mbox{    }+(1-\lambda) \inf_{\delta} E^Q_{t,\tilde a}\left[\frac{1}{2\gamma}\int_t^T \| \delta_r\|^2 dr+f^l(S_T,X_T)\right]\\
  &=\lambda \varphi^l(t,s,x)+(1-\lambda)\varphi^l(t, \tilde s,x),\quad \lambda \in [0,1].
\end{split}
\end{equation*}
Now it is enough to take limits to get the result.
\qed\end{proof}
%


\begin{example}[Forward contract for $n=2$ fuels]
We derive here a more explicit expression for the first term $\zeta(0,a)$ of the asymptotic expansion \eqref{asymptotic} of the price at time zero for a forward contract with two fuels as described in \cite{Aid.10}, with payoff \footnote{Such a payoff, as already noticed, does not satisfy the assumption in Lemma \ref{asymp}. Nonetheless, it clearly belongs to $L^2(Q)$ for all measures $Q \in \mathcal M_E$ and the results in Monoyios \cite{monoyios.12} can still be applied getting the same asymptotic expansion as in \eqref{asymptotic}.}
\begin{equation*}
f(a)=f(s,c,y)=g\left(c^1+c^2-y\right) (s^1 +s^2 \mathbf 1_{\{y-c^1> 0\}}).
\end{equation*}
The assets dynamics are given in \eqref{dynamics2}, where we also assume the seasonality components to be zero for clearness (they would only appear as a mean component in the expressions for the derivatives of $\psi$ below). The no-arbitrage price under the MMM $Q^0$ is
$$p^0(t,a)=E^0_{t,a}[f(A_T)]=\psi^1(t,x)s^1+\psi^2(t,x)s^2$$
where $a=(s,x)$, $s=(s^1,s^2)$, $x=(c^1,c^2,y)$, and
$$\psi^i(t, x)=\int_{\mathbb R^2}\Psi_{C^1_T-D_T}(t,z)\Psi_{C^2_T}(t,c)g(c+z) \chi^i (z) dcdz$$
for $i=1,2$, where we set $$\chi^i (z) := \mathbf 1_{\{z<0\}}+\mathbf 1_{\{z\geq 0,i=1\}}$$ and $\Psi_{C^1_T-D_T}(t,\cdot)$ stands for the conditional density of $C^1_T-D_T$ given $C^1_t=c^1, D_t=y$ (and similarly for $\Psi_{C^2_T}(t,\cdot)$). Notice that an explicit expression for the  price $p^0 (t,a)$ has been obtained in \cite{Aid.10} together with an efficient numerical method to compute it. \\ Based on the previous expression, we can obtain an explicit formula for the derivatives of $p^0(t,a)$ as an intermediate step towards the optimal hedging strategy. We have
$$p^0_x(t,a)=\left(
\begin{array}{c}
\psi^1_{C^1}(t,x)s^1+\psi^2_{C^1}(t,x)s^2\medskip\\
\psi^1_{C^2}(t,x)s^1+\psi^2_{C^2}(t,x)s^2\medskip\\
\psi^1_{D}(t,x)s^1+\psi^2_{D}(t,x)s^2
\end{array}
\right)$$
where
\begin{align*} \psi^i_{C^1}&(t, x)=\frac{e^{-\alpha_{C^1}(T-t)}}{\textrm{Var}_t(C^1_T-D_T)} \int_{\mathbb R^2}(z- c^1e^{-\alpha_{C^1} (T-t)}+ ye^{-\alpha_D (T-t)})\Psi_{C^1_T-D_T}(t,z)\Psi_{C^2_T}(t,c)g(c+z)\chi^i (z) dcdz\\
&\\
\psi^i_{C^2}&(t, x)=\frac{e^{-\alpha_{C^2}(T-t)}}{\textrm{Var}_t(C^2_T)}\int_{\mathbb R^2}(c- c^2e^{-\alpha_{C^2} (T-t)})\Psi_{C^1_T-D_T}(t,z)\Psi_{C^2_T}(t,c)g(c+z) \chi^i (z) dcdz\\
&\\
\psi^i_{D}&(t, x)= -\frac{e^{-\alpha_D (T-t)}}{\textrm{Var}_t(C^1_T-D_T)}\int_{\mathbb R^2}(z- c^1e^{-\alpha_{C^1} (T-t)}+ ye^{-\alpha_D (T-t)})\Psi_{C^1_T-D_T}(t,z)\Psi_{C^2_T}(t,c)g(c+z) \chi^i (z) dcdz
\end{align*}
for $i=1,2$ with $\textrm{Var}_t$ denoting the conditional variance at time $t$, which in our case can be explicitly computed since $C^1 -D$ and $C^2$ are generalized Ornstein-Uhlenbeck processes with time-dependent deterministic coefficients. 
By defining
$$\phi^i(j,x)=\int_0^T e^{\sigma_{i}^2(T-t)} E_{0,x}[\beta^2_{j}\psi^i_{j}(t, X_t)^2] dt , \quad \phi^{12}(j,x)=\int_0^T E_{0,x}[\beta^2_{j}\psi^1_{j}(t, X_t)\psi^2_{j}(t, X_t)] dt,$$
for $i=1,2$ and $j \in \{C^1 , C^2 , D\}$, we finally obtain
\begin{equation*}
\begin{split}
\zeta(0,a)=&\left(\displaystyle\sum_{j\in\{C^1,C^2, D\}}\phi^1(j,x)\right)(s^1)^2+\left(\displaystyle\sum_{j\in\{C^1,C^2,D\}}\phi^2(j,x)\right)(s^2)^2\\
&+\left(\displaystyle\sum_{j\in\{C^1,C^2, D\}}\phi^{12}(j,x)\right)s^1s^2.
\end{split}
\end{equation*}

\begin{remark}
By direct computation as above, one can also obtain similar expressions for spread call options. Pricing spread call options is particularly important in energy markets since such derivatives constitute the building blocks for evaluating the central plants in the real option approach as in, e.g., \cite{Carmona_sch.12} (see the next section for a comparison between UIP and the non-arbitrage MMM price of spread call options).
\end{remark}

\end{example}

\section{Conclusions}
In this paper we considered the utility indifference pricing problem in a particular market model that includes tradable and nontradable assets, and where the derivatives' payoffs possibly depend on both classes. Using BSDE techniques, we established some existence and regularity results for the price, showing in particular how they can be applied to the pricing and hedging of power derivatives under a structural modeling framework. Although we did not aim for the greatest generality we believe that, under suitable assumptions, most of the results could be extended to a broader set of asset dynamics. Nevertheless, we remark that our framework already allows to consider derivatives written on underlyings that possibly exhibit spikes and discontinuities (as it is the case for electricity prices).

\appendix
\section{Auxiliary results and their proofs}

\begin{lemma}   \label{L2}
Let $f \in L^1 (Q_0)$ be bounded from below and let $(Y,Z)$, with $Z=(Z^S,Z^X)$, be a solution to the BSDE \eqref{BSDE_MMM}. Assume that for some $q>0$ there exists a constant $C>0$ such that $\|Z^S_t\|\leq C\|S_t\|^q$ for all $t\in[0,T]$. Then the solution of \eqref{BSDE_MMM} satisfies, for all $p>1$
\begin{equation*}
\begin{split}
  E^0\left[\left(\int_0^{t}\|Z^X_u\|^2du\right)^p\right]\leq CE^0\left[\left(\int_0^t\|\xi_u\|^2 du\right)^{p/2}+1\right]\\
\end{split}
\end{equation*}
where $\xi$ comes from the martingale representation of $f$ under the MMM $Q^0$.
\end{lemma}

\begin{proof} Consider the BSDE \eqref{BSDE_g}
$$Y_t=f+\int_t^T g(Z_r)dr-\int_t^T Z_r dW^0_r$$
write the generator as $g(z)=-\frac{\gamma}{2}\|(0,z^X)\|^2=-\frac{\gamma}{2}\|z\|^2+\frac{\gamma}{2}\|(z^S,0)\|^2$. Notice that $g(Z_r)$ can also be expressed as
\[ g(Z_r) = -\frac{\gamma}{2}\| Z_r \| ^2 + a(t),\]
with $a(t) = \frac{\gamma}{2} \| (Z_t ^S , 0_d) \| ^2 $, which satisfies $|a(t)| \leq C' \| S_t \| ^{2q}$ for some constant $C'>0$.
\\We now assume that $f$ is positive, the case where it is only bounded from below being analogous. Consider the function
$$u(x)=\frac{1}{\gamma^2}(e^{-\gamma x}-1+\gamma x),\quad x\geq 0,$$
from $\mathbb R_+$ to itself. Remark that $u(x)\geq 0$ and $u'(x)\geq 0$ for $x\geq 0$. Moreover, $\gamma u'(x)+u''(x)=1$ and $u(x)\leq \frac{x}{\gamma}$, $u'(x)\leq \frac{1}{\gamma}$, $u''(x)\leq 1$ for $x\geq 0$.
Defining
$$\tau_\kappa=\inf\{t\geq 0:\int_0^t \|Z_u\|^2 du\geq n\} ,\quad \inf \emptyset = +\infty,$$
and applying It\^o's lemma we get
\begin{equation*}
\begin{split}
  u(Y_0)&=u(Y_{t\wedge\tau_\kappa})+\int_0^{t\wedge\tau_\kappa}\left(u'(Y_s)g(Z_s)-\frac{1}{2}u''(Y_s)\|Z_s\|^2\right)ds-\int_0^{t\wedge\tau_\kappa} u'(Y_s)Z_s dW^0_s\\
  &\leq u(Y_{t\wedge\tau_\kappa})+\int_0^{t\wedge\tau_\kappa}u'(Y_s)a(s)-\int_0^{t\wedge\tau_\kappa}\frac{1}{2}\left(\gamma u'(Y_s)+ u''(Y_s)\right)\|Z_s\|^2ds\\
  &\mbox{    }-\int_0^{t\wedge\tau_\kappa} u'(Y_s)Z_s dW^0_s\\
  &=u(Y_{t\wedge\tau_\kappa})+\int_0^{t\wedge\tau_\kappa}u'(Y_s)a(s)-\int_0^{t\wedge\tau_\kappa}\frac{1}{2}\|Z_s\|^2ds-\int_0^{t\wedge\tau_\kappa} u'(Y_s)Z_s dW^0_s\\
\end{split}
\end{equation*}
therefore
\begin{equation*}
    \begin{split}
    \frac{1}{2}\int_0^{t\wedge\tau_\kappa}\|Z_s\|^2ds&\leq u(Y_{t\wedge\tau_\kappa})+\int_0^{t\wedge\tau_\kappa}u'(Y_s)a(s)ds-\int_0^{t\wedge\tau_\kappa} u'(Y_s)Z_s dW^0_s\\
    &\leq Y_{t\wedge\tau_\kappa}+\int_0^{t\wedge\tau_\kappa}u'(Y_s)a(s)ds+\sup_{0\leq t\leq T}\left|\int_0^{t\wedge\tau_\kappa} u'(Y_s)Z_s dW^0_s\right|
    \end{split}
\end{equation*}
and using the Burholder-Davis-Gundy inequalities we obtain
\begin{eqnarray*}
  E^0\left[\left(\int_0^{t\wedge\tau_\kappa}\|Z_s\|^2ds\right)^p\right] &\leq& CE^0\left[Y_{t\wedge\tau_\kappa}^p+\left(\int_0^{t\wedge\tau_\kappa}u'(Y_s)a(s)ds\right)^p\right]\\
  &&+CE^0\left[\left(\int_0^{t\wedge\tau_\kappa} u'(Y_s)^2\|Z_s\|^2 ds\right)^{p/2}\right]\\
  &\leq& CE^0\left[Y_{t\wedge\tau_\kappa}^p+\left(\int_0^{t\wedge\tau_\kappa}u'(Y_s)a(s)ds\right)^p+1\right]\\
  &&+\frac{1}{2}E^0\left[\left(\int_0^{t\wedge\tau_\kappa} \|Z_s\|^2 ds\right)^{p}\right]
\end{eqnarray*}
where we used Young's inequality in the last line. Therefore 
\begin{equation*}
\begin{split}
  E^0\left[\left(\int_0^{t\wedge\tau_\kappa}\|Z_s\|^2ds\right)^p\right]&\leq CE^0\left[\sup_{r\in[0,t]}(E^0_r[f])^p+\left(\int_0^{t}\|S_r\|^2dr\right)^p+1\right]\\
  &\leq CE^0\left[\left(\sup_{r\in[0,t]}\int_0^r\xi_s dW_s\right)^p+1\right]\\
  &\leq CE^0\left[\left(\int_0^t\|\xi_s\|^2 ds\right)^{p/2}+1\right]\\
\end{split}
\end{equation*}
where $\xi$ comes from the martingale representation of $f$ under $Q^0$. The result follows by Fatou's lemma.
\qed\end{proof}

\begin{lemma}   \label{brownian}
Let $W$ be a $\mathbb R^{n+d}$-valued Brownian Motion, $T>0$, $p>1$ and $0<\alpha<p/2$. Define $U_t=\int_0^t u (r) dW_r$, where $u$ is a $\mathbb R^{n+d}$-valued deterministic bounded process.
Then
$$E\left[\sup_{0\leq t\leq T}\frac{|U_t|^p}{t^{\alpha}}\right]<\infty .$$
\end{lemma}

\begin{proof} By Dumbis-Dubins-Schwarz representation of the martingale $U_t$, there exists a Brownian motion $\widetilde W$ such that $U_t = \widetilde W_{\tau_t}$ where $\tau_t = \langle U\rangle _t =  \int_0 ^t \| u(r) \|^2 dr$ is a deterministic bounded time change. Thus, using the scaling property of Brownian motion we have
\begin{equation*}
\begin{split}
E\left[\sup_{0\leq t\leq T}\frac{|U_t|^p}{t^{\alpha}}\right] &= E\left[\sup_{0\leq t\leq T}\frac{|\widetilde W_{\tau_t}|^p}{t^{\alpha}}\right] =E\left[|\widetilde W_{1}|^p\right] \sup_{0\leq t\leq T}\frac{\tau_t^{p/2}}{t^{\alpha}}\\
&\leq C E\left[|\widetilde W_{1}|^p\right] \sup_{0\leq t\leq T}t^{p/2-\alpha}<\infty,
\end{split}
\end{equation*}
for some constant $C>0$. This ends the proof.
\qed\end{proof}

\begin{lemma}	\label{unif_int}
Let $W$ be a $\mathbb R^{n+d}$-valued Brownian motion, $U$ be defined as in Lemma \ref{brownian} and let $K$ be a process in $\mathbb H^{q'} (\mathbb R)$ for some $q'\geq 1$. Suppose, moreover, that $|K_t|\leq F(t,W_t)$ for all $t\in [0,T)$ for some continuous function $F:[0,T)\times\mathbb R^{n+d}\to\mathbb R$. Then there exists $p'>1$ such that
  $$E_t\left[\left(\int_t^T\frac{U_r-U_t}{(r-t)} K_rdr\right)^{p'}\right]<\infty .$$
\end{lemma}
\begin{proof}
We have, by choosing $0<\alpha'<1/2$ and applying H\"older's inequality
\begin{equation*}
\begin{split}
E_t \left[\left(\int_t^T\frac{U_r-U_t}{(r-t)} K_rdr\right)^{p'}\right]
&=E_t\left[\left(\int_t^T\frac{U_r-U_t}{(r-t)^{\alpha'}}\frac{K_r}{{(r-t)^{1-\alpha'}}} dr\right)^{p'}\right]\\
&\leq E_t\left[\left(\sup_{t\leq r\leq T}\frac{|U_r-U_t|}{(r-t)^{\alpha'}}\right)^{p'}\left(\int_t^T\frac{K_r}{{(r-t)^{1-\alpha'}}} dr\right)^{p'}\right]\\
&\leq E_t\left[\left(\sup_{t\leq r\leq T}\frac{|U_r-U_t|}{(r-t)^{\alpha'}}\right)^{pp'}\right]^{1/p}E_t\left[\left(\int_t^T\frac{K_r}{{(r-t)^{1-\alpha'}}} dr\right)^{p'q}\right]^{1/q}\\
&= E_t\left[\sup_{t\leq r\leq T}\frac{|U_r-U_t|^{pp'}}{(r-t)^{pp'\alpha'}}\right]^{1/p}E_t\left[\left(\int_t^T\frac{K_r}{{(r-t)^{1-\alpha'}}} dr\right)^{p'q}\right]^{1/q}\\
&\leq C E_t\left[\left(\int_t^T\frac{K_r}{{(r-t)^{1-\alpha'}}} dr\right)^{p'q}\right]^{1/q}\\
\end{split}
\end{equation*}
by Lemma \ref{brownian}, where the $p >1$ used above is arbitrary.
Now set $p'q=q'$ and recall that $q'>1$ and it can be chosen arbitrarily close to $1$.
Now define
$$\tau=\inf\{r>t: \|W_r-W_t \|\geq M\}, \quad \inf \emptyset =+\infty ,$$
and notice that, for any $0<\varepsilon < T-t$, when $t\leq r\leq \tau\wedge {(T-\varepsilon)}$ we have $|K_r|\leq \tilde M$, where $\tilde M$ is a constant depending on $M$ and on the function $F$. Thus we obtain
\begin{equation*}
\begin{split}
& E_t\left[\left(\int_t^T\frac{K_r}{{(r-t)^{1-\alpha'}}} dr\right)^{q'}\right]\\
&\leq E_t\left[\left(\int_t^{\tau\wedge {(T-\varepsilon)}}\frac{K_r}{{(r-t)^{1-\alpha'}}} dr\right)^{q'}+\left(\int_{\tau\wedge {(T-\varepsilon)}}^T\frac{K_r}{{(r-t)^{1-\alpha'}}} dr\right)^{q'}\right]\\
&\leq C+E_t\left[\left(\int_{\tau\wedge {(T-\varepsilon)}}^T\frac{K_r}{{(r-t)^{1-\alpha'}}} dr\right)^{q'}\right]\\
&\leq C+E_t\left[\frac{1}{{(\tau\wedge {(T-\varepsilon)}-t)^{q'(1-\alpha')}}}\left(\int_{\tau\wedge {(T-\varepsilon)}}^T|K_r| dr\right)^{q'}\right]\\
&\leq C+E_t\left[\frac{1}{{(\tau\wedge {(T-\varepsilon)}-t)^{lq'(1-\alpha')}}}\right]^{1/l}E_t\left[\left(\int_{\tau\wedge {(T-\varepsilon)}}^T|K_r| dr\right)^{q'\frac{l}{1-l}}\right]^{\frac{1-l}{l}}\\
&\leq C+CE_t\left[\frac{1}{{(\tau \wedge {(T-\varepsilon)}-t)^{lq'(1-\alpha')}}}\right]^{1/l}.\\
\end{split}
\end{equation*}
To conclude the proof it suffices to show that the expectation in the RHS of the last inequality is finite. This is a straightforward consequence of Lemma \ref{bessel} below since, conditionally to $\mathcal F_t$, the process $( \|W_{t+u}-W_t \| )_{u\geq 0}$ is clearly a Bessel process of dimension $n+d$ and $lq'(1-\alpha') >1$.
\qed\end{proof}

\begin{lemma}\label{bessel}
Let $R$ be a Bessel process of any positive integer dimension $k\geq 1$ with $R_0 = 0$. Let $\tau_b := \inf \{ t \geq 0 : R_t = b\}$ (with the convention $\inf \emptyset = \infty$) its first hitting time of a level $b>0$.
Then we have that $E[\tau_b ^{-p}] < \infty$ for any $p\geq 1$.
\end{lemma}

\begin{proof} First notice that $t^{-(n+1)} = n! \int_0 ^\infty x^n e^{-tx} dx$ for all $n\geq 0$. Replacing $t$ with $\tau_b$, taking expectations on both sides and using Fubini's theorem, we get
\[ E\left[ \tau_b ^{-(n+1)}\right] = n! \int_0 ^\infty x^n E\left[e^{-x \tau_b}\right] dx .\]
The Laplace transform for the hitting time $\tau_b$ ($b>0$) of a $k$-dimensional Bessel process starting from zero is given by (see, e.g., \cite{going-yor})
\[  E\left[e^{-x \tau_b}\right]  = \left(\frac{x}{2}\right)^{\nu/2} \Gamma^{-1}(\nu+1)\frac{b^\nu}{I_\nu (b\sqrt{2x})},\]
where $\nu = k/2 -1$ is the index of the Bessel process $R$, $\Gamma$ denotes the Gamma function and $I_\nu$ is the modified Bessel function of the first kind of order $\nu$. Thus, to conclude the proof it suffices to show that
\[ \int_0 ^\infty \frac{x^{n+\frac{\nu}{2}}}{I_\nu (b\sqrt{2x})} dx = C \int_0 ^\infty \frac{y^{\nu+1+2n}}{I_\nu (y)}dy < \infty,\]
for a constant $C>0$, which easily follows from the asymptotic behavior of the modified Bessel function $I_\nu (y)$ for small and large $y$ given in \cite{lebedev} (relations 5.16.4 and 5.16.5).
\qed\end{proof}

\begin{lemma}   \label{L2_selling}
Let $f$ be a payoff satisfying Assumption \ref{bounds} with super-replicating portfolio process $V_t := V^{v_1}_t(\pi_1)$ expressed under the MMM $Q^0$ as
$$V_t=\tilde f -\int_t^T L_s dW^{S,0}_s ,\quad \tilde f = V_T ,$$
where $L$ is some adapted process satisfying
$$E^0\left[\left(\int_0^{T}\|L_s\|^2ds\right)^p\right]<\infty$$
for some $p>1$. Then the solution $(Y,Z)$ of \eqref{BSDE_selling} also verifies
\begin{equation*}
\begin{split}
  E^0\left[\left(\int_0^{T}\|Z_s\|^2ds\right)^p\right]<\infty.\\
\end{split}
\end{equation*}
\end{lemma}
\begin{proof}
Define
$$U_t=V_t-Y_t=\tilde f-f+\frac{\gamma}{2}\int_t^T \| Z^X_s\|^2ds-\int_t^T((L_s,0)- Z_s) dW^0_s .$$
Clearly $U_t\geq 0$.
Now if the conditions are satisfied, then following the proof of Lemma \ref{L2} we deduce that
$$E^0\left[\left(\int_0^{t}\|(L_s,0)-Z_s\|^2ds\right)^p\right]\leq C E^0\left[\left(\int_0^t\| L_s\|^2 ds\right)^{p/2}+1\right]$$
for some constant $C$, which implies the result.
\qed\end{proof}


\begin{thebibliography}{9}
\thispagestyle{empty}
\bibitem[ACL10]{Aid.10}
A\"id, R., Campi, L. and Langren\'e, N. (2012), A Structural risk-neutral model for pricing and hedging power derivatives, \emph{Mathematical Finance}. Article first published online: 13 FEB 2012 DOI: 10.1111/j.1467-9965.2011.00507.x

\bibitem[ACLP12]{Aid.12} A\"id, R., Campi, L., Langren\'e, N. and Pham, H. (2012), A probabilistic numerical method for optimal multiple switching problem applied to investment in electricity generation. Preprint, oai:hal.archives-ouvertes.fr:hal-00747229

\bibitem[AID10]{Ankirchner.07}
Ankirchner, S., Imkeller, P. and Dos Reis, G. (2010), Pricing and hedging of derivatives based on nontradable underlyings, \emph{Mathematical Finance}, 20(2), 289-312.

\bibitem[Ba02]{Barlow.02}
Barlow, M. (2002), A diffusion model for electricity prices, \emph{Mathematical Finance}, 12(4), 287-298.

\bibitem[BCK07]{Benth.07}
Benth, F.E., Cartea, A. and Kiesel, R. (2008), Pricing forward contracts in power markets by the certainty equivalence principle: explaining the sign of the market risk premium, \emph{Journal of Banking and Finance}, 32(10), 2006-2021.

\bibitem[Be03]{becherer.03}
Becherer, D. (2003), Rational hedging and valuation of integrated risks under constant absolute risk aversion, \emph{Insurance: Mathematics and economics}, 33(1), 1--28.

\bibitem[Be06]{becherer.06}
Becherer, D. (2006), Bounded solutions to backward SDEs with jumps for utility optimization and indifference hedging, \emph{Annals of Applied Probability}, 16(4), 2027--2054.

\bibitem[BH07]{Briand.07}
Briand, P. and Hu, Y. (2008), Quadratic BSDEs with convex generators and unbounded terminal conditions, \emph{Probab. Theory and Related Fields}, 141(3), 543--567.

\bibitem[CC12]{Carmona.12}
Carmona, R. and Coulon, M. (2012), A survey of commodity markets and structural models for electricity prices, \emph{Proceedings from the special thematic year at the Wolfgang Pauli Institute, Vienna}, forthcoming.

\bibitem[CCS12]{Carmona_sch.12}
Carmona, R. and Coulon, M. and Schwarz, D. (2012), Electricity price modeling and asset valuation: a multi-fuel structural approach, published online in \emph{Mathematics and Financial Economics} DOI 10.1007/s11579-012-0091-4.

\bibitem[CV08]{Cartea.08}
Cartea, A. and Villapiana, P. (2008), Spot price modeling and the valuation of electricity forward contracts: the role of demand and capacity, \emph{Journal of Banking and Finance}, 32,2501-2519.

\bibitem[Co09]{Coulon.09}
Coulon, M. and Howison, S. (2009), Stochastic behavior of the electricity bid stack: from funtamental drivers to power prices, \emph{Journal of Energy Markets}, 2,29-69.

\bibitem[Da97]{Davis.97} \textsc{Davis, M.H.A.} (1997). Option pricing in incomplete markets. In: Mathematics of Derivative Securities, eds M.A.H. Dempster and S.R. Pliska, pages 216-227. Cambridge University Press.

\bibitem[ER00]{nicole.00}
El Karoui, N. and Rouge, R. (2000), Pricing via utility maximization and entropy, \emph{Mathematical Finance}, 10(2), 259--276.


\bibitem[FR75]{fleming.75}
Fleming, W.H. and Rishel, R.W. (1975), Deterministic and stochastic optimal control, Springer Verlag, New York.

\bibitem[FS06]{fleming.06}
Fleming, W.H. and Soner, H.M. (2006), Controlled Markov processes and viscosity solutions, Springer, New York.

\bibitem[FS08]{freischw.08}
Frei, C. and Schweizer, M. (2008), Exponential utility indifference valuation in two Brownian settings with stochastic correlation, \emph{Advances in Applied Probability}, 40(2), 401--423.

\bibitem[GJY03]{going-yor}
Göing-Jaeschke, A., and Yor, M. (2003), A survey and some generalizations of Bessel processes, \emph{Bernoulli}, 9(2), 313--349.

\bibitem[He02]{henderson.02}
Henderson, V. (2002), Valuation of claims on nontraded assets using utility maximization, \emph{Mathematical Finance}, 12(4), 351--373.

\bibitem[HL11]{henderson.11}
Henderson, V. and Liang, G. (2011), A Multidimensional Exponential Utility Indifference Pricing Model with Applications to Counterparty Risk, Preprint downloadable from http://arxiv.org/abs/1111.3856

\bibitem[HH09]{henderson.04}
Henderson, V. and Hobson, D. (2009), Utility indifference pricing: an overview. Chapter 2 of Indifference Pricing: Theory and Applications, ed. R. Carmona, Princeton University Press.

\bibitem[HIM05]{imkeller.05}
Hu, Y., Imkeller, P. and Muller, M. (2005), Utility maximization in incomplete markets, \emph{The Annals of Applied Probability}, 15(3), 1691--1712.

\bibitem[IRR12]{imkeller.12}
Imkeller, P., R\'eveillac, A. and Richter, A.(2012), Differentiability of quadratic BSDEs generated by continuous martingales, \emph{The Annals of Applied Probability}, 22(1), 285--336.

\bibitem[Ho05]{hobson.05}
Hobson, D. (2005), Bounds for the utility-indifference prices of non-traded assets in incomplete markets, \emph{Decisions in Economics and Finance}, 28, 33--52.


\bibitem[Ko00]{Kob.00}
Kobylanski, M. (2000), Backward stochastic differential equations and partial differential equations quadratic growth, \emph{The Annals of Probability}, 28(2), 558--602.

\bibitem[Le72]{lebedev}
Lebedev, N. N. (1972), Special functions and their applications, Dover publications.

\bibitem[MY07]{ma.yong.07} Ma, J. and Yong, J. (2007). \emph{Forward-backward stochastic differential equations and their applications} (Vol. 1702). Springer.

\bibitem[MZ02]{ma.02}
Ma, J. and Zhang, J. (2002), Representation theorems for backward stochastic differential equations, \emph{The Annals Appl. Probab.}, 12, 1390--1418.

\bibitem[Mo12]{monoyios.12}
Monoyios, M. (2012), Malliavin calculus method for asymptotic expansion of dual control problems, Preprint downloadable from http://arxiv.org/abs/1209.6497

\bibitem[OZ09]{OZ.99}
Owen, M. and Zitkovic, G. (2009), Optimal investment with an unbounded random endowment and utility-based pricing, \emph{Mathematical Finance}, 19(1), 129--159.

\bibitem[PP90]{pardoux.peng.90}
Pardoux, E., and Peng, S. (1990), Adapted solution of a backward stochastic differential equation, \emph{Systems \& Control Letters}, 14(1), 55--61.

\bibitem[Ph02]{pham.02}
Pham, H. (2002), Smooth solutions to optimal investment models with stochastic volatilities and portfolio constraints, \emph{Applied Mathematics and Optimization}, 46, 55--78.

\bibitem[PJ08]{pirrong.08}
Pirrong, C., and Jermakyan, M. (2008), The price of power: the valuation of power and weather derivatives, \emph{Journal of Banking and Finance}, 32, 2520--2529.

\bibitem[Sc01]{schweizer.01}
Schweizer, M. (2001), A Guided Tour through Quadratic Hedging Approaches. In: E. Jouini, J. Cvitanic, M. Musiela (eds.), \emph{Option Pricing, Interest Rates and Risk Management}, Cambridge University Press, 538--574.

\bibitem[SZ04]{sircar.zari}
Sircar, R., and Zariphopoulou, T. (2004), Bounds and asymptotic approximations for utility prices when volatility is random, \emph{SIAM Journal on Control and Optimization}, 43(4), 1328--1353.


\bibitem[SGI00]{skantze.00}
Skantze, P., Gubina, A. and Ilic, M. (2000), Bid-based stochastic model for electricity prices: the impact of fundamental drivers on market dynamics, MIT Energy Laboratory Publication MITEL 00-004.

\bibitem[Zh05]{zhang.05}
Zhang, J. (2005), Representation of solutions to BSDEs associated with a degenerate FSDE, \emph{The Annals Appl. Probab.}, 15, 1798--1831


\end{thebibliography}
\end{document}